\documentclass[11pt]{article}

\usepackage{amssymb,amsthm,amsbsy}

\theoremstyle{definition}
\newtheorem{remark}{Remark}

\theoremstyle{plain}

\usepackage{nath1}
\usepackage{hyperref}

\def\eqref#1{{\rm(\ref{#1})}}

\def\lin{\partial}

\def\bA{\mathbf A}

\def\bB{\mathbf B}
\def\bM{\mathbf M}

\def\bb{\boldsymbol\beta}
\def\bW{\mathbf W}
\def\bV{\mathbf V}
\def\bc{\mathbf c}

\def\br{\mathbf r}
\def\be#1{\mathbf e_{#1}}
\def\bE{\mathbf E}
\def\bF{\mathbf F}
\def\bL{\mathbf L}
\def\bP{\mathbf P}
\def\bQ{\mathbf Q}
\def\bR{\mathbf R}
\def\bS{\mathbf S}

\def\bU{\mathbf U}

\def\diag{\delta}

\def\nd{\mathsf{n.d.}}

\newcommand{\mg}{\mathfrak{g}}
\newcommand{\pd}{{\partial}}
\newcommand{\cprime}{\/{\mathsurround=0pt$'$}}

\newcommand{\gl}{\mathfrak{gl}}
\newcommand{\ve}{\varepsilon}

\newcommand{\nv}{m}
\newcommand{\al}{\alpha}

\newcommand{\sm}{d}
\newcommand{\fzc}{\mathcal{A}}
\newcommand{\nsy}{\mathbf{\Gamma}}

\newcommand{\la}{{\lambda}}
\newcommand{\fik}{\mathbb{K}}
\newcommand{\Com}{\mathbb{C}}

\newcommand{\twl}{\mathfrak{T}}
\newcommand{\tgr}{\mathcal{G}}

\newcommand{\tpa}{\mathbb{L}}

\newcommand{\vf}{\varphi}

\newcommand{\sla}{\mathfrak{L}}
\newcommand{\cen}{\mathfrak{Z}}
\newcommand{\pot}{P}
\newcommand{\bet}{\beta}

\newcommand{\zp}{\mathbb{Z}_{\ge 0}}
\newcommand{\zpin}{\mathbb{Z}_{\ge 0}\cup\infty}
\newcommand{\zsp}{\mathbb{Z}_{>0}}

\newcommand{\beq}{\begin{equation}}
\newcommand{\ee}{\end{equation}}
\newcommand{\lb}{\label}
\newcommand{\er}{\eqref}
\newcommand{\tr}{\mathop{\mathrm{tr}}\nolimits}

\newtheorem{proposition}{Proposition}

\begin{document}
\title{On construction of symmetries and recursion operators from  
zero-curvature representations and the Darboux--Egoroff system}

\author{Sergei Igonin \\
\footnotesize Department of Mathematics, Utrecht University, \\
\footnotesize P.O. Box 80010, 3508 TA Utrecht, the Netherlands \\
\footnotesize E-mail: s-igonin@yandex.ru
\\[2ex]
Michal Marvan \\
\footnotesize Mathematical Institute, Silesian University in Opava, \\
\footnotesize Na Rybn\'\i\v cku 1, 746 01 Czech Republic \\
\footnotesize E-mail: Michal.Marvan@math.slu.cz}

\date{}

\maketitle

\abstract{
The Darboux--Egoroff system of PDEs with any number $n\ge 3$ of independent variables plays an essential role in the problems of describing $n$-dimensional flat diagonal metrics of Egoroff type and Frobenius manifolds. We construct a recursion operator and its inverse for symmetries of the Darboux--Egoroff system and describe some symmetries generated by these operators. 
   
The constructed recursion operators are not pseudodifferential, but are B\"acklund autotransformations for the linearized system whose solutions correspond to symmetries of the Darboux--Egoroff  system. For some other PDEs, recursion operators of similar types were considered previously by Papachristou, Guthrie, Marvan, Pobo\v{r}il, and Sergyeyev. 
   
In the structure of the obtained third and fifth order symmetries of the Darboux--Egoroff system, one finds the third and fifth order flows of an $(n-1)$-component vector modified KdV hierarchy. 
   
The constructed recursion operators generate also an infinite number of nonlocal symmetries. In particular, we obtain a simple construction of nonlocal symmetries that were studied by Buryak and Shadrin in the context of the infinitesimal version of the Givental--van de Leur twisted loop group action on the space of semisimple Frobenius manifolds. 
   
We obtain these results by means of rather general methods, using only the zero-curvature representation of the considered PDEs.}

\smallskip

\bigskip
\noindent{\it Keywords\/}: 
symmetries of PDEs; recursion operators; zero-curvature representations; the Darboux-Egoroff system; 
vector modified KdV hierarchy; twisted loop algebras

\medskip
\noindent 2000 \emph{AMS Mathematics Subject Classification}:
37K10, 37K35, 37K30

\section{Introduction}
\lb{sec-intr}

Symmetries, recursion operators for symmetries, and 
parameter-dependent zero-curvature representations belong 
to the main tools in the theory of integrable PDEs, 
but the interrelations between these structures are not fully understood. 

Consider a PDE system with independent variables $x_1,\dots,x_n$ 
and dependent variables $u^1,\dots,u^{\nv}$ 
\beq
\lb{intf0}
F^p\big(x_i,u^j,u^j_{x_a},
u^j_{x_{a}x_{b}},\dots\big)=0,\qquad 
u^j=u^j(x_1,\dots,x_n),\qquad  
j=1,\dots,\nv,\qquad  
p=1,\dots,s.
\ee
Here $u^j_{x_a},\,u^j_{x_{a}x_{b}},\dots$ denote the partial derivatives of $u^j$. 
Suppose that one has a zero-curvature representation (ZCR) 
\beq
\lb{intzcr}
\bA^{k}_{x_l}-\bA^{l}_{x_k}+[{\bA}^{k},{\bA}^{l}]=0,\qquad 
\bA^{k}=\bA^{k}(\la,x_i,u^j,u^j_{x_a},\,u^j_{x_{a}x_{b}},\dots),
\qquad k,l=1,\dots,n,
\ee
where $\la$ is a parameter, $\bA^{k}$ take values in a matrix Lie algebra, and 
equations~\eqref{intzcr} hold as a consequence of~\eqref{intf0}.

Equations~\eqref{intzcr} are equivalent to the compatibility of 
the corresponding auxiliary linear system $\Psi_{x_k}=\bA^{k}\Psi$, $k=1,\dots,n$.  
Such a ZCR often helps to establish integrability 
for system~\eqref{intf0} in the framework of the inverse scattering method.

The present paper concerns a particular case of the following problem. 
If a parameter-dependent matrix-valued ZCR~\eqref{intzcr} is given for system~\eqref{intf0}, 
how to construct recursion operators and symmetries for~\eqref{intf0}?

We include in this question local and nonlocal symmetries. 
As is well known, nonlocal symmetries often 
generate infinite-dimensional (non-abelian) Lie algebras, 
which reflect the algebraic structures behind the integrability properties 
of the considered PDEs. 

Before describing the results of the paper, 
we would like to recall some facts on ZCRs and recursion operators. 
Note that equations~\eqref{intf0},~\eqref{intzcr} 
may have any number $n\ge 2$ of independent variables $x_1,\dots,x_n$.
According to~\cite{marvan-zcr92}, 
if one has a nontrivial ZCR~\eqref{intzcr} with $n\ge 3$  
then system~\eqref{intf0} must be overdetermined in a certain sense. 

In the case $n=2$, for $(1+1)$-dimensional PDEs with independent variables $x$ and $t$, 
a typical recursion operator is a pseudodifferential operator 
mapping symmetries to (possibly nonlocal) symmetries, 
thus capable of generating infinite series of them 
(see, e.g.,~\cite{akns,Ol,fokas87,olv_eng2,K,kerst-kras,gerdjikov}). 
Such a pseudodifferential recursion operator is usually written  
as a linear combination of products of terms of the form 
$f$, $\,\,D_x$, $\,\,D_x^{-1}$.  
Here $D_x$ is the total derivative in~$x$, 
and $f$ is a (matrix) function on the corresponding jet bundle. 

For example, recall that, for each nonnegative integer $k$, 
the $(2k+1)$th order flow in the KdV hierarchy can be written as 
$u_{t_{2k+1}}=\pd_x^{2k+1}u+g_{2k+1}(u,u_x,\dots,\pd_x^{2k-1}u)$ for some function $g_{2k+1}$. 
In particular, $u_{t_1}=u_x$ and $u_{t_3}=u_{xxx}+uu_x$. 
These flows are symmetries of the KdV equation $u_t=u_{xxx}+uu_x$, 
where $t$ can be identified with $t_3$. 
The classical recursion operator $D_x^2+\frac23u+\frac13u_xD_x^{-1}$ of KdV maps 
the $(2k+1)$th order flow to the $(2k+3)$th order flow. So the whole hierarchy can be generated 
by applying this operator repeatedly to~$u_x$.

For a wide class of ZCRs with $n=2$,   
there are a variety of methods to construct a recursion operator from a ZCR   
as a pseudodifferential operator of this type  
(see, e.g., \cite{akns,G-K-S,K,sakovich-zcr} and references therein). 
It seems that this approach does not immediately extend to the case of arbitrary~$n$. 
Furthermore, since pseudodifferential recursion operators 
are defined as formal expressions containing $D_x^{-1}$, 
sometimes it is not clear how to apply such an operator to a given symmetry~\cite{G,S-W}, 
because $D_x^{-1}(h)$ is not always well defined for functions $h$ on jet bundles. 

A different method to construct recursion operators from ZCRs~\eqref{intzcr} with arbitrary~$n$ 
was suggested in~\cite{M-S,M-Einstein,M-P}, using some ideas of~\cite{G}. 
More precisely, in~\cite{M-S,M-Einstein} one considered the case $n=2$, 
and in~\cite{M-P} the method of~\cite{M-S,M-Einstein} was applied 
to an overdetermined PDE with any number $n\ge 3$ of independent variables.   

Recursion operators obtained by the method of~\cite{M-S,M-Einstein,M-P} are not pseudodifferential, 
but are B\"acklund autotransformations of a certain type.
Namely, as is well known, for a given system~\eqref{intf0}, one can write down the linearized 
system, whose solutions correspond to symmetries of~\eqref{intf0}.
Then the recursion operators from~\cite{M-S,M-Einstein,M-P} can be viewed as B\"acklund autotransformations of the linearized system.

It seems that Papachristou~\cite{Pap} was the first to interpret some recursion operators 
as B\"acklund autotransformations of linearized equations.  
This approach is widely applicable (Guthrie~\cite{G}, Marvan~\cite{M2}) 
and does not suffer from the lack of a rigorous definition 
of the action of pseudodifferential operators on some types of symmetries~\cite{G,M2}. 
If one regards a recursion operator as a B\"acklund autotransformation, 
then it can be applied to any symmetry, 
but the result may depend on pseudopotentials in the sense of~\cite{we_prol1975,G}.
Also, in this approach, the inverse of a recursion operator 
can be computed effectively in many cases~\cite{G,M2,M3,M-S,M-Einstein,M-P}.  

The main idea of~\cite{M-S,M-Einstein,M-P} 
is to search for a recursion operator depending on a matrix pseudopotential $\bW$ 
defined by the compatible equations $\bW_{x_k}=[{\bA}^{k}, {\bW}] + \lin_{\bU} \bA^{k}$, 
$k=1,\dots,n$, where ${\bA}^{k}$ are the matrices from a ZCR~\eqref{intzcr} 
and $\lin_{\bU}$ is the evolutionary vector field corresponding 
to a symmetry~$\bU$ (see Section~\ref{sec-prelim} for more details).
Considered examples~\cite{M3,M-S,M-Einstein,M-P} show that very often such a 
recursion operator can be found relatively easily, 
but it generates mainly nonlocal symmetries. 
However, by inverting it, one obtains another recursion operator, 
which often generates hierarchies of local symmetries. 

Note that a recursion operator does not immediately give 
an infinite-dimensional Lie algebra of (possibly nonlocal) symmetries. 
In general, in the terminology of~\cite{vinbook,kerst-kras,K-V2,M2}, 
a recursion operator produces nonlocal symmetry shadows, which may depend 
on some nonlocal variables called pseudopotentials~\cite{we_prol1975,G}.
In order to define the commutators of such symmetry shadows, 
one needs to specify a suitable action of them on the corresponding pseudopotentials, 
which is not always possible. 

In some cases this is possible, and then the obtained nonlocal symmetries generate 
a Lie algebra, which often turns out to be infinite-dimensional 
(see, e.g.,~\cite{guth-hick} for the KdV equation). 
Note that the problem of constructing Lie algebras of nonlocal symmetries was not 
considered in~\cite{M-S,M-Einstein,M-P}.
 
The general definitions of symmetries, symmetry shadows, 
and nonlocal symmetries for PDEs are recalled in 
Section~\ref{sec-prelim} of the present paper. 
Also, in Section~\ref{sec-prelim} we review the method of~\cite{M-S,M-Einstein,M-P} 
for constructing recursion operators from ZCRs. 

\begin{remark}
According to~\cite{marvan-zcr92}, 
existence of a nontrivial matrix-valued ZCR~\eqref{intzcr} in the case $n\ge 3$ 
imposes strong conditions on the possible form of system~\eqref{intf0}. 
For example, 
the Kadomtsev-Petviashvili (KP) equation 
$(u_t+u_{xxx}+6uu_x)_x=\pm u_{yy}$ 
does not possess any nontrivial matrix-valued ZCRs, 
so the approach of~\cite{M-S,M-Einstein,M-P} and the present paper 
cannot produce any recursion operators for the KP equation.

\end{remark}

In Sections~\ref{secinvrec}--\ref{secnonlsym}, 
we study the above-mentioned topics for the Darboux--Egoroff (DE) system
\beq
\lb{intdesys}
\bet_{ij}=\bet_{ij}(x_1,\dots,x_n),\qquad \bet_{ij}=\bet_{ji},\qquad 
i,j=1,\dots,n,\qquad i\neq j,\qquad n\ge 3,\\
\frac{\pd}{\pd x_k} \bet_{ij}=\bet_{ik} \bet_{kj},\qquad i\neq j\neq k\neq i,\quad\qquad
\sum_{s=1}^n \frac{\pd}{\pd x_s}\bet_{ij} = 0.
\ee
Equations~\eqref{intdesys} originate from a classical problem 
of describing flat diagonal metrics of Egoroff type~\cite{Dar,Ego} 
(see also~\cite{dubrovin90,tsarev,zakhduke} for a modern exposition). 
The DE system~\eqref{intdesys} has attracted a lot of attention in the last two decades, 
because, according to the results of Dubrovin~\cite{dubrovin92,dubrovin96}, 
it plays an essential role in the classification problem 
for massive topological field theories and a related theory of Frobenius manifolds, 
which have found remarkable applications in various areas~\cite{dubrovin92,dubrovin96,manin-book}. 

Let $\bb$ be the symmetric $n\times n$ matrix with the entries $\bet_{ij}$, where $\bet_{ii}=0$. 
For $k=1,\dots,n$, we denote by $\be{k}$ the $n \times n$ matrix whose entries are zero 
except for one $1$ occupying the $k$th position on the diagonal. 

It is well known that system~\eqref{intdesys} possesses the following ZCR
\beq
\lb{zcrdeint}
\bA^{k}_{x_l}-\bA^{l}_{x_k}+[{\bA}^{k},{\bA}^{l}]=0,\qquad 
\bA^{k}=\lambda\be{k}+[\bb,\be{k}],\qquad k,l=1,\dots,n,
\ee
where $\lambda$ is a parameter. 
As usual, equations~\eqref{zcrdeint} imply that the auxiliary linear system
\beq
\lb{alsint}
\Psi_{x_k} = \bA^{k} \Psi=(\lambda \be{k} + [\bb,\be{k}])\Psi, \qquad \det \Psi \neq 0,
\qquad k=1,\ldots,n,
\ee
for an $n\times n$ matrix-function $\Psi$ is compatible.

The DE system~\eqref{intdesys} is known to be integrable by the inverse scattering, 
dressing, and algebro-geometric methods
(see, e.g.,~\cite{dubrovin90,krich,tsarev,zakhduke,dubrovin96,vdlmar,vdltw,aglz}), 
but we could not find in the literature any recursion operators for it. 
One of the goals of the present paper is to construct a recursion operator for~\eqref{intdesys} 
and to describe some symmetries generated by this operator. 

More precisely, we construct two recursion operators $\mathfrak{R}$ and~$\hat{\mathfrak{R}}$, 
which are inverse to each other in the sense of~\cite{G,M2,M3}. 
Note that $\mathfrak{R}$, $\hat{\mathfrak{R}}$ differ considerably from 
classical pseudodifferential recursion operators of $(1+1)$-dimensional PDEs.
In accordance with the approach of~\cite{M-S,M-Einstein,M-P}, 
the operators $\mathfrak{R}$, $\hat{\mathfrak{R}}$ 
are not pseudodifferential, but are B\"acklund autotransformations 
for the linearized DE system, whose solutions correspond to symmetries of~\eqref{intdesys}. 

A B\"acklund autotransformation for the DE system~\eqref{intdesys} is known 
(see~\cite{tsarev} and references therein).
We construct the recursion operators $\mathfrak{R}$, $\hat{\mathfrak{R}}$ 
as B\"acklund autotransformations for the \emph{linearized} DE system, which is a very different result.  
To avoid confusion in the terminology, 
we remark that the linearized DE system 
(given by equations~\eqref{Bij},~\eqref{lDE} in Section~\ref{secinvrec}) 
is not related to the auxiliary linear system~\eqref{alsint}.  

The constructed operators $\mathfrak{R}$, $\hat{\mathfrak{R}}$ 
generate local and nonlocal symmetries for system~\eqref{intdesys}. 
The structure of these symmetries is discussed below. 

We obtain these results by means of rather general methods, 
using only the ZCR~\eqref{zcrdeint}.
(We use also the corresponding auxiliary linear system~\eqref{alsint}, 
but it follows immediately from the ZCR.) 
Another goal of the paper is to demonstrate a technique for obtaining such results.
We expect that this technique can be applied successfully to many more PDEs possessing 
parameter-dependent matrix-valued ZCRs. 

The construction consists of three parts. 
\begin{enumerate}
\item 
Using the method of~\cite{M-S,M-Einstein,M-P} for constructing recursion operators from ZCRs, 
in Section~\ref{secinvrec} we obtain a recursion operator~$\mathfrak{R}$ 
for system~\eqref{intdesys}. 
Since $\mathfrak{R}$ generates mainly nonlocal symmetries, 
we call it the \emph{inverse recursion operator} for the DE system~\eqref{intdesys}. 

The operator $\mathfrak{R}$ is closely related to the 
inverse recursion operator~\cite{M-P} of the intrinsic generalized sine-Gordon (IGSG) 
equation~\cite{Am,T-T,CLT,M-P} 
in $n$ independent variables, because system~\eqref{intdesys} 
can be regarded as a certain reduction of the IGSG equation 
(see Section~\ref{secinvrec} for more details). 

\item Using Guthrie's approach~\cite{G} to inverting recursion operators 
by means of pseudopotentials, 
in Section~\ref{secdirect} we invert the operator $\mathfrak{R}$ and 
obtain what we call the \emph{direct recursion operator} $\hat{\mathfrak{R}}$.

Starting from the zero symmetry and applying the operator $\hat{\mathfrak{R}}$ repeatedly, 
one gets higher symmetries for system~\eqref{intdesys}. 
The obtained third and fifth order symmetries are local.
In the main nonlinear part of these symmetries of the DE system~\eqref{intdesys},  
we find the third and fifth order flows of an  
$(n-1)$-component vector modified KdV hierarchy from~\cite{A-F,An,B-M}. 
Note that we do not prove locality  
for higher order symmetries generated by~$\hat{\mathfrak{R}}$. 

Applying $\hat{\mathfrak{R}}$ to the scaling symmetry of system~\eqref{intdesys}, 
one obtains some nonlocal 
symmetries\footnote{In the terminology of~\cite{vinbook,kerst-kras,K-V2,M2},  
these are nonlocal symmetry shadows, but in the literature on recursion operators such 
objects are usually called nonlocal symmetries.}, 
which are briefly discussed in the end of Section~\ref{secdirect}.

Alternatively, one can obtain the operator $\hat{\mathfrak{R}}$ as a reduction 
of the direct recursion operator of the IGSG equation~\cite{M-P},
using the above-mentioned fact that system~\eqref{intdesys} 
can be regarded as a reduction of IGSG.

\item 
Let $\Psi$ satisfy~\eqref{alsint}. 
Introducing a formal Taylor series expansion $\Psi(\la)=\sum_{i=0}^\infty\la^i \Psi_i$ 
for $\Psi$ in $\lambda$ and using~\eqref{alsint}, 
one gets an infinite set of pseudopotentials (nonlocal variables) for system~\eqref{intdesys}. 
In other words, one obtains an infinite-dimensional covering of~\eqref{intdesys} 
in the sense of~\cite{vinbook,K-V2}.
Pseudopotentials arising from such formal Taylor series expansions are often used in 
the study of nonlocal symmetries for integrable PDEs 
(cf., e.g.,~\cite{guth-hick} for the KdV equation).

Applying the inverse recursion operator $\mathfrak{R}$ to the zero symmetry, we obtain 
nonlocal symmetries depending on these pseudopotentials. 
More precisely, in the terminology of~\cite{vinbook,K-V2,M2}, 
application of~$\mathfrak{R}$ to the zero symmetry gives nonlocal symmetry shadows. 
After imposing a certain algebraic constraint on~$\Psi(\la)$,  
which is suggested by the structure of the ZCR, 
we obtain nonlocal symmetries from these shadows 
by specifying a suitable action of them on the pseudopotentials.

As is shown in Section~\ref{secnonlsym}, 
the constructed nonlocal symmetries generate an infinite-dimensional Lie algebra, 
which is isomorphic to the negative part of a twisted loop algebra of~$\gl_n$ divided by its center. 
To prove this isomorphism, we use some results of Buryak and Shadrin~\cite{bur-shadr}, 
who considered essentially the same nonlocal symmetries in the context 
of the Givental--van de Leur twisted loop group action on the space of semisimple Frobenius manifolds
(see Remark~\ref{remburshad} below).

The above-mentioned constraint on~$\Psi(\la)$ reads $\Psi(\la)\cdot\Psi^\top(-\la)=\mathrm{Id}$, 
where $\top$ denotes matrix transposition.
This constraint was used previously in~\cite{vdltw}. 

\end{enumerate}

\begin{remark}
It is known that the Darboux--Egoroff system~\eqref{intdesys} 
can be obtained by a certain reduction of a multi-KP hierarchy~\cite{vdltw,vdlmar,aglz}. 
To our knowledge, this fact does not give a description 
of all possible symmetries of system~\eqref{intdesys}, especially if one considers also nonlocal symmetries.
Indeed, applying a reduction to a PDE, one may get some extra symmetries, 
which are difficult to predict in advance. 

In particular, we do not see any straightforward way to deduce the structure of recursion operators for 
symmetries of~\eqref{intdesys} from known relations of system~\eqref{intdesys} with multi-KP hierarchies. 
As has been discussed above, we construct the recursion operators $\mathfrak{R}$ and $\hat{\mathfrak{R}}$ 
by means of a different approach.

\end{remark}

\begin{remark}
\lb{remburshad}
Let $\twl$ be the infinite-dimensional Lie algebra of Laurent polynomials of the form 
$$
\sum_{i=-p}^q\la^i\br_i,\qquad p,q\in\zp,\qquad 
\br_i^\top=(-1)^{i+1}\br_i,
$$
where $\br_i$ belong to the Lie algebra $\gl_n$ of $n\times n$ matrices. 

Then $\twl$ can be identified with the twisted loop algebra of $\gl_n$ corresponding to the automorphism 
$$
\sigma\colon\gl_n\to\gl_n,\quad\qquad\sigma(\br)=-\br^\top.
$$ 
The algebra $\twl$, as a vector space, is equal to the direct sum of the subalgebras 
$$
\twl_{+}=\left.\{\sum_{i=0}^q\la^i\br_i\ \ \right|\ \ q\ge 0,\ \ \ \br_i^\top=(-1)^{i+1}\br_i\big\},\qquad\quad 
\twl_{-}=\left.\big\{\sum_{i=-p}^{-1}\la^i\br_i\ \ \right|\ \ p>0,\ \ \ \br_i^\top=(-1)^{i+1}\br_i\big\}.
$$

Using the infinitesimal version of the Givental--van de Leur twisted loop group action on the space of semisimple Frobenius manifolds~\cite{fvdls,givental1,givental2,lee1,lee2,vdltw}, 
in the framework of multi-KP hierarchies and Sato's semi-infinite Grassmannian, 
Buryak and Shadrin~\cite{bur-shadr} presented an action of the Lie algebra~$\twl$ 
by nonlocal symmetries of the Darboux--Egoroff system~\eqref{intdesys}. 
(In~\cite{bur-shadr} these symmetries are called infinitesimal deformations 
of solutions of the Darboux--Egoroff system.)

This action of~$\twl$ is described also in Section~\ref{secnonlsym} of the present paper. 
Since the subalgebra $\twl_{+}$ acts by obvious gauge symmetries, 
the interesting part of the action of~$\twl$ is concentrated in the subalgebra~$\twl_{-}$. 
In our construction, nonlocal symmetries generating the action of $\twl_{-}$ 
are obtained by means of a relatively simple procedure, 
which involves application of the inverse recursion operator~$\mathfrak{R}$ to the zero symmetry. 
In particular, we do not use multi-KP hierarchies and Grassmannians.  

The kernel of the action of~$\twl_{-}$ is equal to the center of~$\twl_{-}$, 
so the corresponding algebra of nonlocal symmetries is isomorphic to the algebra~$\twl_{-}$ 
divided by its center (see Section~\ref{secnonlsym} for details).

\end{remark}

\section{Conventions and notation}
\lb{convnot}

Let $\fik$ be either $\Com$ or $\mathbb{R}$. 
In what follows, unless otherwise specified, 
variables and functions take values in~$\fik$. 
Also, the entries of all considered matrices belong to~$\fik$. 
Functions are assumed to be smooth if $\fik=\mathbb{R}$ and analytic if $\fik=\Com$. 

The symbols $\zsp$ and $\zp$ denote the sets of positive and nonnegative 
integers respectively.  
When we write $\nv\in\zpin$, we mean that either $\nv\in\zp$ or $\nv=\infty$.

We shall often use the {\it diagonalization operator} $\diag$, 
which sets all off-diagonal elements of a square matrix to zero.
Otherwise said, if $\bP$ is a square matrix 
then $\diag\bP$ is the diagonal matrix that possesses the same diagonal as $\bP$.

Also, we define the operator $\nd$ as follows $\nd\bP=\bP-\diag\bP$. 
That is, the off-diagonal elements of the matrix $\nd\bP$ are equal to the 
corresponding elements of~$\bP$, and the diagonal of~$\nd\bP$ is zero. 
The notation $\nd$ is taken from~\cite{bur-shadr}. 
It is easy to check that
$$
\numbered\label{diag}
[\diag\bP,\diag\bQ] = 0,\qquad\quad\diag[\diag\bP, \bQ] = 0 
$$
for any square matrices $\bP$, $\bQ$ of the same size. 

\section{Preliminaries on symmetries and recursion operators}
\lb{sec-prelim}

In this section we review some basic notions from 
the theory of symmetries~\cite{olv_eng2,vinbook} and 
nonlocal symmetries~\cite{K-V1,K-V2,vinbook} for PDEs. 
Also, in the second half of the section we recall some results of~\cite{M-S,M-Einstein,M-P} 
on relations between zero-curvature representations and recursion operators for symmetries. 

Let $n\in\zsp$ and $\nv,s\in\zpin$. 
Consider a PDE system with independent variables $x_1,\dots,x_n$ 
and dependent variables $u^1,\dots,u^{\nv}$ 
\beq
\lb{prfxu}
F^\al\big(x_i,u^j_I)=0,\quad\qquad
\al=1,\ldots,s.
\ee
Here $I=(i_1,\ldots,i_k)$ ranges over unordered $k$-tuples of integers 
$i_1,\ldots,i_k\in\{1,\dots,n\}$ 
for all $k\ge 0$, and 
$$
u^j_I=\frac{\pd^{k}u^j}{\pd x_{i_1}\ldots\pd x_{i_k}},\qquad\quad j=1,\dots,\nv.
$$ 
In particular, for $I=\varnothing$ one has $u^j_{\varnothing}=u^j$.

As usual in the formal theory of PDEs~\cite{olv_eng2,vinbook}, 
$x_i$ and $u^j_I$ are regarded as independent
quantities and can be viewed as coordinates on an abstract infinite-dimensional space, 
which is called the \emph{infinite jet space}.
When we consider a scalar function $G(x_i,u^j_I)$, we always assume that it may depend only 
on a finite number of these coordinates. 
In particular,
for each $\al$ the function $F^\al\big(x_i,u^j_I)$ in~\eqref{prfxu} 
depends only on a finite number of $x_i$, $u^j_I$.

The \emph{total derivative operator with respect to} $x_i$ is given by the formula
$$
D_{x_i} = \frac{\pd}{\pd x_i}
 + \sum_{j=1}^{\nv} \sum_{|I|=0}^{\infty}u^j_{Ii}\frac{\pd}{\pd u^j_{I}}.
$$
For $I=(i_1,\ldots,i_k)$, we denote $D_I=D_{x_{i_1}}\cdots D_{x_{i_k}}$. 
Then $D_I(F^\al)=0$ are differential consequences of system~\eqref{prfxu}. 

In what follows, when we write ``$G=0$ modulo~\eqref{prfxu}''  
for some (vector) function $G$, we mean that the equation $G=0$ is valid modulo 
differential consequences of system~\eqref{prfxu}. 

A \emph{symmetry} of system~\eqref{prfxu} is given by an $\nv$-component
vector-function 
\beq
\lb{locsym}
\bU=(U^1,\ldots,U^\nv),\quad\qquad U^p=U^p(x_i,u^j_{I}),\qquad p=1,\ldots,\nv,
\ee
such that the following property holds. 
If $u^1,\ldots,u^\nv$ obey~\eqref{prfxu} then the 
infinitesimal deformation $\tilde u^p=u^p+\ve U^p$, $p=1,\ldots,\nv$,  
satisfies equations~\eqref{prfxu} up to $O(\varepsilon^2)$. 

More precisely, 
this condition means that the functions $U^p$ obey the following equations
\beq
\lb{prlineq}
\sum_{p=1}^{\nv}\sum_{|I|=0}^{\infty}
D_{I}(U^p)\frac{\pd F^\al}{\pd u^p_{I}}=0
\quad\text{modulo~\eqref{prfxu}}\qquad\forall\,\al=1,\ldots,s.
\ee
Equations~\eqref{prlineq} are obtained by substituting $\tilde u^p=u^p+\ve U^p$ 
in place of $u^p$ in~\eqref{prfxu} and collecting the terms linear in~$\ve$. 
So~\eqref{locsym} is a symmetry of~\eqref{prfxu} 
iff equations~\eqref{prlineq} are valid modulo differential consequences of~\eqref{prfxu}. 
Note that~\eqref{prlineq} is sometimes called the \emph{linearized system}  
corresponding to~\eqref{prfxu}.

Let 
\beq
\lb{evvf}
\lin_{\bU}=
\sum_{p=1}^{\nv}\sum_{|I|=0}^{\infty}D_{I}(U^p)\frac{\pd}{\pd u^p_{I}}.
\ee
The operator $\lin_{\bU}$ is called 
the \emph{evolutionary vector field} 
(or the \emph{linearization operator}) corresponding to $\bU$. 
Then equations~\eqref{prlineq} can be written as 
$\lin_{\bU}(F^\al)=0$ modulo~\eqref{prfxu} for $\al=1,\dots,s$.

Symmetries of a given system~\eqref{prfxu} form a Lie algebra, 
where the Lie bracket of $\bU$ and $\bU'$ is defined as follows 
$[\bU,\bU']=\lin_{\bU}(\bU')-\lin_{\bU'}(\bU)$. 
Symmetries~\eqref{locsym} are sometimes called \emph{local symmetries} of~\eqref{prfxu}. 
(In contrast to nonlocal symmetries, which are discussed below.)

Let $N\in\zpin$. Consider a PDE system of the form 
\beq
\lb{covpde}
w^q_{x_k}=G^q_k(w^r,x_i,u^j_{I}),\qquad k=1,\dots,n,\qquad 
q,r=1,\dots,N, 
\ee
where $w^q$ are additional dependent variables and $w^q_{x_k}=\pd w^q/\pd x_k$.
Note that no derivatives of $w^q$ appear on the right-hand side of equations~\eqref{covpde}. 

When we consider a scalar function $H(w^r,x_i,u^j_I)$, 
we assume that it may depend only on a finite number of $w^r$, $x_i$, $u^j_I$. 
In particular, this assumption applies to the functions 
$G^q_k(w^r,x_i,u^j_{I})$ in~\eqref{covpde}.

Equations~\eqref{covpde} are said to be \emph{compatible modulo~\eqref{prfxu}} if 
\beq
\lb{cross}
\sum_r G^r_l\frac{\pd G^q_k}{\pd w^r}+
\frac{\pd G^q_k}{\pd x_l} + \sum_{j,I}u^j_{Il}\frac{\pd G^q_k}{\pd u^j_{I}}
=
\sum_r G^r_k\frac{\pd G^q_l}{\pd w^r}+\frac{\pd G^q_l}{\pd x_k}
+\sum_{j,I}u^j_{Ik}\frac{\pd G^q_l}{\pd u^j_{I}}
\quad\text{modulo~\eqref{prfxu}}\qquad\forall\,q,k,l.
\ee 
So~\eqref{cross} must hold modulo differential consequences of~\eqref{prfxu}. 
Condition~\eqref{cross} is the equation 
$\pd w^q_{x_k}/\pd x_l=\pd w^q_{x_l}/\pd x_k$, 
where we substitute $w^q_{x_k}=G^q_k(w^r,x_i,u^j_{I})$ and 
$w^q_{x_l}=G^q_l(w^r,x_i,u^j_{I})$ according to~\eqref{covpde}. 

In the terminology of~\cite{we_prol1975,G}, 
if equations~\eqref{covpde} are compatible modulo~\eqref{prfxu} then 
$w^q$ are \emph{pseudopotentials} for~\eqref{prfxu}. 

Sometimes one needs to impose additional constraints of the form 
\beq
\lb{constr}
C^a(w^r)=0,\qquad a=1,\dots,Q,\qquad Q\in\zpin,
\ee
where $C^a(w^r)$ depends on a finite number of the variables $w^r$, 
$r=1,\dots,N$.
We say that constraints~\eqref{constr} are \emph{compatible with~\eqref{covpde} modulo~\eqref{prfxu}} 
if 
\beq
\lb{cacomp}
\sum_r w^r_{x_k}\frac{\pd C^a}{\pd w^r}=
\sum_r G^r_k\frac{\pd C^a}{\pd w^r}=0\quad 
\text{modulo~\eqref{prfxu},~\eqref{constr}}\qquad\forall\,k,a.
\ee
We mean here that~\eqref{cacomp} holds modulo equations~\eqref{constr} and differential 
consequences of~\eqref{prfxu}. 
In~\eqref{cacomp} one uses the fact that $w^r_{x_k}=G^r_k$ in view of~\eqref{covpde}.

If equations~\eqref{covpde},~\eqref{constr} are compatible modulo~\eqref{prfxu}, 
then we say that~\eqref{covpde},~\eqref{constr} determine a \emph{covering} of~\eqref{prfxu} 
with pseudopotentials $w^q$. 
A geometric theory of coverings of PDEs can be found in~\cite{vinbook,K-V2}.

When a covering~\eqref{covpde},~\eqref{constr} is given, 
we extend the action of the total derivatives $D_{x_k}$ to functions of $w^q$ by the rule 
$D_{x_k}(w^q)=w^q_{x_k}=G^q_k$. 
That is, for a function $f=f(w^q,x_i,u^j_{I})$ one has 
\beq
\lb{dxkwq}
D_{x_k}(f)=\frac{\pd f}{\pd x_k} + \sum_{j,I} 
u^j_{Ik}\frac{\pd f}{\pd u^j_{I}}+\sum_q 
G^q_k(w^r,x_i,u^j_{I})\frac{\pd f}{\pd w^q}.
\ee

Then one can consider equations~\eqref{prlineq} in the case when $U^p$ may depend on $w^q$. 
In the terminology of~\cite{vinbook,K-V2}, 
a \emph{symmetry shadow} for~\eqref{prfxu} in the covering~\eqref{covpde},~\eqref{constr} 
is an $\nv$-component vector-function 
\beq
\lb{symsh}
\bU=(U^1,\ldots,U^\nv),\qquad U^p=U^p(w^r,x_i,u^j_{I}),\qquad p=1,\ldots,\nv,
\ee
satisfying~\eqref{prlineq}.  

By definition, a \emph{nonlocal symmetry} for~\eqref{prfxu} in the 
covering~\eqref{covpde},~\eqref{constr} 
is a symmetry of system~\eqref{prfxu},~\eqref{covpde},~\eqref{constr}.
That is, a nonlocal symmetry is given by functions 
\beq
\lb{upwq}
U^p=U^p(w^r,x_i,u^j_{I}),\qquad 
W^q=W^q(w^r,x_i,u^j_{I}),\qquad p=1,\ldots,\nv,\qquad 
q=1,\dots,N,
\ee
such that, if $u^p$ obey~\eqref{prfxu} and $w^q$ obey~\eqref{covpde},~\eqref{constr},   
then the infinitesimal deformation 
\beq
\lb{prdeform}
\tilde u^p=u^p+\ve U^p,\qquad \tilde w^q=w^q+\ve W^q,\qquad p=1,\ldots,\nv,\qquad 
q=1,\dots,N,
\ee
satisfies equations~\eqref{prfxu},~\eqref{covpde},~\eqref{constr} up to $O(\varepsilon^2)$. 
This means that $U^p$ obey~\eqref{prlineq}, and $U^p$, $W^q$ obey
\beq
\lb{lcovpde}
D_{x_k}(W^q)=\sum_r W^r\frac{\pd G^q_k}{\pd w^r}+
\sum_{p=1}^{\nv}\sum_{|I|=0}^{\infty}D_{I}(U^p)\frac{\pd G^q_k}{\pd u^p_{I}}
\quad\text{modulo~\eqref{prfxu},~\eqref{constr}}\qquad\forall\,k,q,\\
\sum_r W^r\frac{\pd C^a}{\pd w^r}=0
\quad\text{modulo~\eqref{prfxu},~\eqref{constr}}\qquad\forall\,a.
\ee
Equations~\eqref{lcovpde} are the linearized version of~\eqref{covpde},~\eqref{constr}.

Note that if~\eqref{upwq} is a nonlocal symmetry then ${(U^1,\ldots,U^\nv)}$ 
is a symmetry shadow. 
A coordinate-independent definition of (nonlocal) symmetries 
can be found in~\cite{vinbook}.

Symmetries and symmetry shadows are often constructed 
by means of so-called recursion operators. 
As has been discussed in Section~\ref{sec-intr}, one can find in the literature 
a number of different approaches to this topic. 
Below we describe a method from~\cite{M-S,M-Einstein,M-P} 
for constructing recursion operators from zero-curvature representations.

Suppose that system~\eqref{prfxu} possesses a zero-curvature representation (ZCR)
\beq
\lb{zcrprel}
D_{x_l}(\bA^{k}) - D_{x_k}(\bA^{l}) + [{\bA}^{k},{\bA}^{l}] = 0
\quad\text{modulo~\eqref{prfxu}}\qquad\forall\,k,l=1,\ldots,n,
\ee
where $\bA^{k}=\bA^{k}(x_i,u^j_I)$, $k=1,\dots,n$, 
are functions with values in the algebra of $\sm\times\sm$ matrices for some $\sm\in\zsp$. 

Let $\bU=(U^1,\ldots,U^\nv)$ be a symmetry of system~\eqref{prfxu}.  
Formula~\eqref{evvf} implies $[\lin_{\bU},D_{x_i}]=0$ for all $i=1,\dots,n$. 
Recall that equations~\eqref{prlineq} can be written as 
$\lin_{\bU}(F^\al)=0$ modulo~\eqref{prfxu}.
Since $[\lin_{\bU},D_{x_i}]=0$, we get also 
$$
\lin_{\bU}(D_{x_{i_1}}\cdots D_{x_{i_r}}(F^\al))=
D_{x_{i_1}}\cdots D_{x_{i_r}}(\lin_{\bU}(F^\al))=0
\quad\text{modulo~\eqref{prfxu}}\qquad\forall\,i_1,\dots,i_r\in\{1,\dots,n\}. 
$$

Therefore, if an equation is valid modulo differential consequences of system~\eqref{prfxu}, 
we can apply the operator $\lin_{\bU}$ to this equation. 
Applying $\lin_{\bU}$ to~\eqref{zcrprel}, one gets 
\beq
\lb{linuzcr}
D_{x_l}(\lin_{\bU}\bA^{k}) - D_{x_k}(\lin_{\bU}\bA^{l}) + 
[\lin_{\bU}{\bA}^{k},{\bA}^{l}]+[{\bA}^{k},\lin_{\bU}{\bA}^{l}] = 0
\quad\text{modulo~\eqref{prfxu}}\,\ 
\quad\forall\,k,l,
\ee
where $\lin_{\bU}{\bA}^{k}$ is computed componentwise.  
Using~\eqref{zcrprel} and~\eqref{linuzcr}, it is easy to check that the system 
\begin{equation}
\label{Wgen}
\bW_{x_k} = [{\bA}^{k}, {\bW}] + \lin_{\bU} \bA^{k},\qquad\quad
k=1,\dots,n, 
\end{equation}
for a $\sm\times\sm$ matrix-function $\bW$ is compatible modulo~\eqref{prfxu}. 
Hence $\bW$ can be regarded as a matrix pseudopotential for~\eqref{prfxu}, 
and we can set $D_{x_k}(\bW)=[{\bA}^{k}, {\bW}] + \lin_{\bU} \bA^{k}$. 

Let $W_{ab}$ be the entries of the matrix $\bW$. 
Assume that we have found $\nv$ linear combinations 
\beq
\lb{tilup}
\tilde U^p=\sum_{a,b=1}^{\sm} c^{p,ab}W_{ab}+\sum_{r=1}^{\nv}\hat{c}^{p,r}U^r,\qquad p=1,\ldots,\nv,
\ee
of $W_{ab}$ and~$U^r$ so that the following property holds. 
If $\bU=(U^1,\ldots,U^\nv)$ satisfies~\eqref{prlineq} and the matrix $\bW$ 
obeys~\eqref{Wgen}, then $\tilde{\bU}=(\tilde{U}^1,\ldots,\tilde{U}^\nv)$ 
given by~\eqref{tilup} satisfies~\eqref{prlineq} as well.

We assume that this property holds for any $\bU=(U^1,\ldots,U^\nv)$ satisfying~\eqref{prlineq}. 
In particular, $\bU$ may depend on some pseudopotentials $w^q$, 
so $\bU$ is a symmetry shadow. 

The coefficients $c^{p,ab}$, $\hat{c}^{p,r}$ in~\eqref{tilup} 
may depend on $x_i$, $u^j_{I}$ 
(see also Remark~\ref{remcoef} below for more general possibilities). 
If the ZCR~\eqref{zcrprel} depends on a parameter~$\la$,
then $c^{p,ab}$, $\hat{c}^{p,r}$ may also depend on $\la$. 

Thus, if $\bU=(U^1,\ldots,U^\nv)$ is a symmetry shadow and the matrix $\bW$ 
obeys~\eqref{Wgen}, then $\tilde{\bU}=(\tilde{U}^1,\ldots,\tilde{U}^\nv)$ 
given by~\eqref{tilup} is a symmetry shadow as well. 
According to~\eqref{tilup}, if $\bU$ depends on $x_i$, $u^j_{I}$, and 
some pseudopotentials $w^q$,  
then $\tilde{\bU}$ may depend on $x_i$, $u^j_{I}$, and the pseudopotentials $w^q$, $W_{ab}$.

Then the correspondence $\bU \mapsto \tilde{\bU}$ 
is a recursion operator in Guthrie's sense~\cite{G} for system~\eqref{prfxu}.

The correspondence $\bU \mapsto \tilde{\bU}$ 
can be viewed as a B\"acklund autotransformation 
for the linearized system~\eqref{prlineq}. 
Indeed, we take a solution $\bU$ of~\eqref{prlineq} and construct a new solution $\tilde{\bU}$  
by means of pseudopotentials~\eqref{Wgen}, so this is similar to classical 
B\"acklund autotransformations. 
See~\cite{M2,M3} for the general geometric theory of recursion 
operators as B\"acklund autotransformations of linearized equations. 

In the studied examples~\cite{M-S,M3,M-Einstein,M-P}, 
the right-hand side of~\eqref{tilup} tends to be a simple expression. 

The examples from~\cite{M-S,M3,M-Einstein,M-P} 
show that usually a recursion operator $\mathfrak{R}$ constructed in this way 
generates symmetry shadows depending nontrivially 
on pseudopotentials (so these shadows are not local symmetries). 
Often, in order to obtain local symmetries~\eqref{locsym}, 
one needs to invert $\mathfrak{R}$ by means of a procedure described in~\cite{G,M3}.
We demonstrate this procedure in Section~\ref{secdirect} 
in the case of the Darboux--Egoroff system. 

\begin{remark}
\lb{remcoef}
The described setting is sufficient for the present paper 
and the examples studied in~\cite{M-S,M-Einstein,M-P}. 
In principle, one can consider also more general recursion operators such that  
the coefficients $c^{p,ab}$, $\hat{c}^{p,r}$ in~\eqref{tilup} may depend on some pseudopotentials. 
Moreover, one may add to the right-hand side of~\eqref{tilup} 
terms of the form $D_{I}(U^r)$ with some coefficients, but usually this is not necessary.
\end{remark}

\begin{remark}
\lb{remwpsi} 
For any symmetry $\bU$, equations~\er{zcrprel} imply that the following system 
\beq
\label{psibv}
\Psi_{x_k} = \bA^{k} \Psi, \qquad \det \Psi \neq 0,\qquad
\bV_{x_k} = \Psi^{-1}(\lin_{\bU} \bA^{k})\Psi,
\qquad k=1,\ldots,n,
\ee
for $n\times n$ matrix-functions $\Psi$, $\bV$ is compatible.

If $\Psi$, $\bV$ obey~\er{psibv} then $\bW = \Psi \bV \Psi^{-1}$ satisfies~\eqref{Wgen}. 
Indeed, for $\bW = \Psi \bV \Psi^{-1}$ we have 
$$
\bW_{x_k} = (\Psi \bV \Psi^{-1})_{x_k}  
 = \Psi_{x_k} \bV \Psi^{-1}
   + \Psi \bV_{x_k} \Psi^{-1}
   - \Psi \bV \Psi^{-1} \Psi_{x_k} \Psi^{-1}
\\\quad
 = \Psi_{x_k} \Psi^{-1} \bW 
   + \lin_{\bU} \bA^{k}
   - \bW \Psi_{x_k} \Psi^{-1}
 = \bA^{k} \bW 
   + \lin_{\bU} \bA^{k}
   - \bW \bA^{k}=[{\bA}^{k}, {\bW}] + \lin_{\bU} \bA^{k}.
$$ 

Therefore, one can use the entries $W_{ab}$ of the matrix $\bW = \Psi \bV \Psi^{-1}$ 
in the recursion operator~\er{tilup}. 
An example of this construction is given in Proposition~\ref{prop:iROtriv} 
in Section~\ref{secinvrec}.
\end{remark}

\begin{remark}
If the functions $G^q_k$ in~\eqref{covpde} do not depend on $w^r$, 
then the pseudopotentials $w^q$ are sometimes called \emph{potentials}.
\end{remark}

\section{The inverse recursion operator}
\lb{secinvrec}

Using the method of~\cite{M-S,M-Einstein,M-P} 
for constructing recursion operators from zero-curvature representations (ZCRs), 
in the present section we shall obtain a recursion operator for the Darboux--Egoroff system. 
Also, we shall describe some (nonlocal) symmetry shadows generated by this operator. 

\begin{remark}
\lb{remdeigsg}

Using the above-mentioned method, Marvan and Pobo\v{r}il~\cite{M-P} 
constructed a recursion operator for the 
intrinsic generalized sine-Gordon (IGSG) equation~\cite{Am,T-T,CLT,M-P}. 
In principle, a recursion operator for the Darboux--Egoroff system 
can be deduced from that for IGSG, 
because the Darboux--Egoroff system is a special reduction of IGSG. 

Indeed, in~\cite{M-P} the IGSG equation is written as a system of PDEs with independent variables 
$x^1,\dots,x^n$, dependent variables $h^{ij}$, $v^i$, $i,j=1,\dots,n$, and a parameter $K$. 
Assuming $K=0$ and $h^{ij}=h^{ji}$, 
the IGSG system decouples, and the part containing only the unknowns $h^{ij}$ 
becomes the Darboux--Egoroff system (if we replace $h^{ij}$ by $\bet_{ij}$ 
and $x^i$ by $x_i$). 
However, a recursion operator for Darboux--Egoroff is not immediately obvious 
from that for IGSG, because the original $2n \times 2n$ ZCR  
of the IGSG equation does not immediately reduce to the $n \times n$ 
ZCR~\eqref{zcrdeint} of the Darboux--Egoroff system.

In our opinion, it is more instructive to present the application 
of the above-mentioned method to the Darboux--Egoroff system in full detail, 
because this may help the readers to apply the method to other PDEs possessing ZCRs.

Furthermore, we shall describe the structure of some (nonlocal) symmetry shadows 
generated by the obtained recursion operator, and this will help us to construct 
nonlocal symmetries for the Darboux--Egoroff system in Section~\ref{secnonlsym}. 
Note that nonlocal symmetries and symmetry shadows were not studied in~\cite{M-P}. 

\end{remark}

In what follows, subscripts after a comma mean that we apply (total) derivatives 
with respect to the corresponding variables. 
For example, $\bet_{ij,k}=\pd \bet_{ij}/\pd x_k$.  
In this notation, the Darboux--Egoroff system~\eqref{intdesys} reads
\begin{equation}
\label{DE}
\bet_{ij,k} = \bet_{ik} \bet_{kj},\qquad i\neq j\neq k\neq i,\qquad i,j,k=1,\dots,n,\\
\sum_{s=1}^n\bet_{ij,s} = 0,\qquad \bet_{ij}=\bet_{ji}.
\end{equation}

Recall that we denote by $\bb$ 
the symmetric $n\times n$ matrix with the entries $\bet_{ij}$, where $\bet_{ii}=0$.
Then system~\eqref{DE} can be compactly written as
\begin{equation}
\label{DEm}
[\bb,\be{k}]_{,l} - [\bb,\be{l}]_{,k} + [[\bb,\be{k}], [\bb,\be{l}]] = 0,\qquad\quad 
k,l=1,\dots,n,
\end{equation}
where $\be{k}$ is the $n \times n$ matrix whose entries are zero 
except for one $1$ occupying the $k$th position on the diagonal. 
To show that~\eqref{DE} is equivalent to~\eqref{DEm}, one can use the fact that~\eqref{DE} implies 
\beq
\lb{bijjiji}
\bet_{ij,j}+\bet_{ij,i}+\sum_{s\neq i,j}\bet_{is} \bet_{sj}=0,
\qquad\quad i\neq j.
\ee

According to the general definition of symmetries from Section~\ref{sec-prelim}, 
a symmetry of system~\eqref{DE} is given by functions 
\beq
\lb{Bij}
B_{ij},\qquad i\neq j,\qquad B_{ij}=B_{ji},\qquad i,j=1,\dots,n,
\ee
such that if $\bet_{ij}$ satisfy~\eqref{DE} then $\tilde \bet_{ij}=\bet_{ij}+\ve B_{ij}$ 
obey~\eqref{DE} up to $O(\ve^2)$. 

This means that $B_{ij}$ satisfy the {\it linearized Darboux--Egoroff system}
\begin{equation}
\label{lDE}
B_{ij,k} = \bet_{ik} B_{kj} + \bet_{kj} B_{ik}, \qquad i\neq j\neq k\neq i,\qquad 
\sum_{s=1}^n B_{ij,s} = 0.
\end{equation}
Equations~\eqref{lDE} are obtained by substituting $\tilde \bet_{ij}=\bet_{ij}+\ve B_{ij}$ 
in place of $\bet_{ij}$ in~\eqref{DE} and collecting the terms linear in~$\ve$.

Let $\bB$ be the symmetric matrix with the entries $B_{ij}$, where $B_{ii}=0$. 
Then the linearized Darboux--Egoroff system~\eqref{lDE} can be written as 
\begin{equation}
\label{lDEm}
[\bB,\be{k}]_{,l} - [\bB,\be{l}]_{,k}
+[[\bB,\be{k}], [\bb,\be{l}]] + [[\bb,\be{k}], [\bB,\be{l}]] = 0,\qquad\quad 
k,l=1,\dots,n.
\end{equation}
Indeed, substituting $\tilde{\bb}=\bb+\ve\bB$ in place of $\bb$ in~\eqref{DEm}  
and collecting the terms linear in~$\ve$, one gets~\eqref{lDEm}.

According to the definitions of symmetries and symmetry shadows discussed 
in Section~\ref{sec-prelim}, for $\bB$ satisfying~\eqref{lDEm} we have the following. 
If $\bB$ depends on $x_i$, $\bet_{ij}$, and derivatives of $\bet_{ij}$, then $\bB$ is a symmetry. 
If $\bB$ depends also on some pseudopotentials, then $\bB$ is a symmetry shadow. 

However, 
in the literature on recursion operators, symmetry shadows are very often called symmetries as well. 
We shall also sometimes use this abuse of terminology, when it does not lead to a confusion. 
In particular, in the present section and in Section~\ref{secdirect} a symmetry of the 
Darboux--Egoroff system is any solution $\bB$ 
of the linearized Darboux--Egoroff system~\eqref{lDEm}, 
where $\bB$ may depend on pseudopotentials. 

The ZCR~\eqref{zcrdeint} of the Darboux--Egoroff system can be written as
\beq
\lb{zcr}
\bA^{k}_{,l}-\bA^{l}_{,k}+[{\bA}^{k},{\bA}^{l}]=0,\qquad 
\quad k,l=1,\dots,n, 
\ee
where $\bA^{k}=\lambda\be{k}+[\bb,\be{k}]$.  
It is easy to check that equations~\eqref{zcr} 
hold as a consequence of system~\eqref{DE} (or~\eqref{DEm}), irrespectively
of the value of the parameter $\lambda$. 

Analogously to~\eqref{Wgen}, we introduce an $n \times n$ matrix ${\bW} = (W_{ij})$ satisfying
\begin{equation}
\label{W}
\bW_{,k} = [{\bA}^{k}, {\bW}] + \lin_{\bB} \bA^{k},\qquad\quad k=1,\dots,n.
\end{equation}
Here $\lin_{\bB}$ is the corresponding evolutionary vector field 
(also called the linearization operator), which is defined similarly to~\eqref{evvf}. 
One has $\lin_{\bB}(\bet_{ij,k_1 \dots k_l})=B_{ij,k_1 \dots k_l}$, 
and $\lin_{\bB} \bA^{k}$ is computed componentwise. 
Compatibility of system \eqref{W} follows from the zero-curvature 
condition~\eqref{zcr}, as has been discussed in Section~\ref{sec-prelim} for system~\eqref{Wgen}. 

Since $\lin_{\bB}(\be{k})=0$ and $\lin_{\bB}(\bb)=\bB$, 
for $\bA^{k}=\lambda\be{k}+[\bb,\be{k}]$ we get $\lin_{\bB} \bA^{k} = [\bB,\be{k}]$. 
Hence~\eqref{W} is equivalent to
$$
\numbered\label{WH}
\bW_{,k} = [{\bA}^{k}, {\bW}] + [\bB,\be{k}],\qquad\quad k=1,\dots,n. 
$$
With $i,j,k$ denoting pairwise different indices, equations~\eqref{WH} read 
$$
\numbered\label{iROw}
W_{ii,i} = -\sum_s \bet_{is} (W_{is} + W_{si}), \qquad
W_{ii,k} = \bet_{ik} (W_{ik} + W_{ki}), \\
W_{ij,i} = \lambda W_{ij} + \bet_{ij} W_{ii} - \sum_s \bet_{is} W_{sj} - B_{ij}, \\
W_{ij,j} = -\lambda W_{ij} + \bet_{ij} W_{jj} - \sum_s \bet_{sj} W_{is} + B_{ij}, \qquad 
W_{ij,k} = \bet_{ik} W_{kj} + \bet_{kj} W_{ik}.
$$

\begin{proposition} 
\label{prop:iRO}
If\/ $\bB=(B_{ij})$ is a symmetry of the Darboux--Egoroff system and $W_{ij}$ 
satisfy~\eqref{iROw}, then $\tilde{\bB}=(\tilde{B}_{ij})$ given by $\tilde{B}_{ii}=0$,
\begin{equation}
\label{iRO}
\tilde{B}_{ij} = W_{ij} + W_{ji}, \qquad i\neq j,\qquad i,j=1,\dots,n,
\end{equation}
is a symmetry of the Darboux--Egoroff system as well.
\end{proposition}

That is, $\tilde{B}_{ij}$ for $i\neq j$ are off-diagonal components of the symmetric matrix 
$\bW + \bW^\top$.

Recall that in Section~\ref{sec-prelim} we have discussed recursion operators 
of the form~\eqref{tilup}. 
Similarly to~\eqref{tilup}, Proposition~\ref{prop:iRO} says that formula~\eqref{iRO} 
determines a recursion operator for the Darboux--Egoroff system. 

Recursion operators of this type are often called inverse, 
because they generate mainly nonlocal symmetries. 
So we say that~\eqref{iRO} is the \emph{inverse recursion operator} for the Darboux--Egoroff system 
and denote it by $\mathfrak{R}$.

To simplify the proof of Proposition~\ref{prop:iRO}, let us rewrite its statement in
terms of the symmetric and antisymmetric components
$$
\bS = \bW^\top + \bW, \quad\qquad \bR = \bW^\top - \bW,
$$
where ${}^\top$ denotes transposition. 
Substituting $\bA^{k}=\lambda\be{k}+[\bb,\be{k}]$ in~\eqref{WH}, one obtains 
$$
\bW_{,k} = [{\bA}^{k}, {\bW}] + [\bB,\be{k}] 
= \lambda [\be{k},\bW] + [[\bb,\be{k}], \bW] + [\bB,\be{k}].
$$
By transposition,
$\bW_{,k}^\top 
 = -\lambda [\be{k},\bW^\top] + [[\bb,\be{k}], \bW^\top] - [\bB,\be{k}]$. 
By addition and subtraction, equation~\eqref{WH} is equivalent to the system
$$
\numbered\label{RS}
\bS_{,k} = -\lambda [\be{k}, \bR] + [[\bb,\be{k}], \bS],\qquad\quad
\bR_{,k} = -\lambda [\be{k}, \bS] + [[\bb,\be{k}], \bR] - 2 [\bB,\be{k}].
$$
Hence the following equivalent form of Proposition~\ref{prop:iRO} in
matrix notation.

\begin{proposition} 
\label{prop:iROm}
Let\/ $\bB$ be a symmetry of the Darboux--Egoroff system.
Let\/ $\bR$ be an antisymmetric, $\bS$ a symmetric matrix satisfying system~\eqref{RS}.
Then the off-diagonal part
$$
\numbered\label{H'}
\tilde{\bB} = \bS - \diag \bS
$$
of\/ $\bS$ is a symmetry of the Darboux--Egoroff system as well. 
Here $\diag$ is the diagonalization operator defined in Section~\textup{\ref{convnot}}. 
\end{proposition}

\begin{proof} 
Let $\bB$ satisfy the linearized Darboux--Egoroff system~\eqref{lDEm}.
To verify that $\tilde{\bB} = \bS - \diag \bS$ satisfies equations~\eqref{lDEm} as well, we 
compute the left-hand side
$$
\bL = [\tilde{\bB},\be{k}]_{,l} - [\tilde{\bB},\be{l}]_{,k}
+ [[\tilde{\bB},\be{k}], [\bb,\be{l}]] + [[\bb,\be{k}], [\tilde{\bB},\be{l}]]
$$
and show that it simplifies to zero.
Since every two diagonal matrices commute, we have $[\diag\bS,\be{k}] = 0$. 
It follows that all terms containing $\diag \bS$ vanish and
$$
\bL\wall = [\bS_{,l},\be{k}]
 - [\bS_{,k},\be{l}] + [[\bS,\be{k}], [\bb,\be{l}]]
 + [[\bb,\be{k}], [\bS,\be{l}]]\\ 
= -\lambda [[\be{l}, \bR],\be{k}] + [[[\bb,\be{l}], \bS],\be{k}]
 + \lambda [[\be{k}, \bR],\be{l}] - [[[\bb,\be{k}], \bS],\be{l}]
 + [[\bS,\be{k}], [\bb,\be{l}]] + [[\bb,\be{k}], [\bS,\be{l}]]. 
\return
$$
Using the Jacobi identity and $[\be{k},\be{l}] = 0$, one routinely verifies 
that $\bL = 0$.
\end{proof}

To gain a better insight into the inverse recursion operator, 
we are going to use the construction described in Remark~\ref{remwpsi}. 
To this end, we consider the auxiliary linear system
\beq
\label{Q}
\Psi_{,k} = \bA^{k} \Psi, \qquad \det \Psi \neq 0,
\qquad k=1,\ldots,n,
\ee
for an $n\times n$ matrix $\Psi$.
Compatibility of system~\eqref{Q} follows from the zero-curvature condition~\eqref{zcr}.

\begin{proposition}
\label{prop:iROtriv}
Let\/ $\bB$ be a symmetry of the Darboux--Egoroff system, 
and let an $n \times n$ matrix $\bV$ satisfy
$$
\numbered\label{V}
\bV_{,k} = \Psi^{-1} [\bB,\be{k}] \Psi,\qquad k=1,\ldots,n.
$$

Set\/ $\bW = \Psi \bV \Psi^{-1}$.
Then 
\beq
\lb{bhbw}
\tilde{\bB} = \bW + \bW^\top - 2 \diag \bW
\ee
is another symmetry of the Darboux--Egoroff system. 
\end{proposition} 
\begin{proof}
Since $\lin_{\bB} \bA^{k} = [\bB,\be{k}]$, equation~\er{V} 
can be written as $\bV_{,k} = \Psi^{-1}(\lin_{\bB} \bA^{k})\Psi$. 
According to Remark~\ref{remwpsi}, equations~\er{Q} and 
$\bV_{,k} = \Psi^{-1}(\lin_{\bB} \bA^{k})\Psi$ imply that 
the matrix $\bW = \Psi \bV \Psi^{-1}$ satisfies~\eqref{W}. 
The rest is Proposition~\ref{prop:iRO} translated to matrix notation.
\end{proof}

As has been said in Section~\ref{convnot}, for any square matrix $\bP$ we denote 
$\nd\bP=\bP-\diag\bP$. Then formula~\eqref{bhbw} can be written as $\tilde{\bB}=\nd(\bW + \bW^\top)$. 

Let us apply Proposition~\ref{prop:iROtriv} to the zero symmetry $\bB=0$. 
Equation~\eqref{V} becomes simply
$\bV_{,k} = 0$, whence $\bV= \bc$ and $\bW= \Psi \bc \Psi^{-1}$, 
where $\bc$ is an arbitrary constant $n\times n$ matrix. Therefore, 
\beq
\lb{bhmain}
\tilde{\bB}=\nd(\bW + \bW^\top)=\nd(\Psi\bc \Psi^{-1} + (\Psi\bc\Psi^{-1})^\top).
\ee
In other words,~\eqref{bhmain} is the off-diagonal part of the matrix 
$\Psi\bc \Psi^{-1} + (\Psi\bc\Psi^{-1})^\top$.

As has been discussed above, in this section, symmetry shadows are also called symmetries. 
Using the more precise terminology introduced in Section~\ref{sec-prelim}, 
we can say that, for every constant $n\times n$ matrix~$\bc$, 
the matrix~\eqref{bhmain} is a symmetry shadow in the covering determined by 
the compatible system~\eqref{Q}.

\section{The direct recursion operator}
\lb{secdirect}

By inverting the inverse recursion operator constructed in Section~\ref{secinvrec}, 
we shall obtain what we call the direct recursion operator.  
The inversion means that we express the preimage
(formerly $\bB$, newly $\tilde{\bB}$) in terms of the image (formerly $\tilde{\bB}$, newly $\bB$). 
Under the notational change indicated, the relevant equations~\eqref{RS} and~\eqref{H'} 
read
$$
\numbered\label{iROm}
\bB = \bS - \diag \bS,
\\
\bS_{,k} = - \lambda [\be{k},\bR] + [[\bb,\be{k}],\bS],
\\
\bR_{,k} = - \lambda [\be{k},\bS] + [[\bb,\be{k}],\bR] - 2 [\tilde{\bB},\be{k}].
$$
The task is to express $\tilde{\bB}$ in terms of $\bB$, $\bb$, and some pseudopotentials.

To start with, we find the diagonal $\bP = \diag\bS$ by applying the diagonalization 
operator $\diag$ on both sides of the second equation from~\eqref{iROm}.
Since $\be{k}$ is diagonal, $\diag[\be{k},\bR] = 0$ by~\eqref{diag}, and similarly
$\diag[[\bb,\be{k}],\diag\bS] = 0$.
Therefore,
$$
\numbered\label{P}
\bP_{,k} = \diag \bS_{,k} = \diag[[\bb,\be{k}],\bS] = \diag[[\bb,\be{k}],\bB + \diag\bS]
 = \diag[[\bb,\be{k}],\bB],
$$
which is a compatible system of equations by virtue of~\eqref{DEm},~\eqref{lDEm}. 
Hence the entries of the diagonal matrix $\bP$ can be regarded as pseudopotentials. 

In what follows we shall need the useful formula
$$
\numbered\label{rec}
\bF - \diag\bF = \frac12 \sum_k [[\bF, \be{k}], \be{k}],
$$
which holds for an arbitrary $n \times n$ matrix $\bF$.
Consequently, the off-diagonal part $\bF - \diag\bF$ can be 
reconstructed from the commutators $[\bF, \be{k}]$. 
For instance, the third equation of system~\eqref{iROm} immediately yields
$$
[2 \tilde{\bB} - \lambda \bS, \be{k}] = -\bR_{,k} + [[\bb,\be{k}],\bR],
$$
where the diagonal part of $2 \tilde{\bB} - \lambda \bS$ is $-\lambda\,\diag \bS$.
Hence, formula~\eqref{rec} allows us to compute 
$2 \tilde{\bB} - \lambda (\bS - \diag \bS)$.
Combining the result with the first equation from~\eqref{iROm}, we obtain
$$
\numbered\label{iROm3}
4\tilde{\bB}-2\lambda\bB
 = -\sum_k [\bR_{,k},\be{k}] + \sum_k [[[\bb,\be{k}],\bR],\be{k}]
$$
(note that the right-hand side is symmetric, because $\bR$ is antisymmetric by 
assumption).

Next, from the second equation of system~\eqref{iROm} we have
\beq
\lb{labrbe}
\la[\bR,\be{k}]=\bS_{,k}-[[\bb,\be{k}],\bS].
\ee
As $\bR$ is antisymmetric, one has $\diag\bR=0$. 
Using formula~\eqref{rec} for $\bF=\la\bR$, from~\eqref{labrbe} we get 
$$
\numbered\label{R}
2\la\bR=\sum_k [\bS_{,k},\be{k}]-\sum_k [[[\bb,\be{k}],\bS],\be{k}]
$$
(note that the right-hand side is antisymmetric, because $\bS$ is symmetric by assumption).

Since $\bP$ and $\bP_{,k}$ are diagonal, we have 
$[\bP,\be{l}] = 0 = [\bP_{,k},\be{l}]$ for all $k,l$.
Then
$$
\sum_k [[[\bb,\be{k}],\bP],\be{k}] = \sum_k [[[\bb,\be{k}],\be{k}],\bP]
 = 2\,[\bb,\bP]
$$
by formula~\eqref{rec}, since $\diag \bb = 0$.
Therefore, inserting $\bS = \bB + \diag\bS = \bB + \bP$ into~\eqref{R}, we obtain 
$$
\numbered\label{RP}
2\lambda \bR = \sum_k [\bB_{,k},\be{k}]
 - \sum_k [[[\bb,\be{k}],\bB],\be{k}]
 - 2\,[\bb,\bP].
$$  

Multiplying~\eqref{iROm3} by $-2\la$ and replacing $k$ by $l$, one gets 
$$
\numbered\label{iROm3la}
-8\la\tilde{\bB}+4\lambda^2\bB
 = \sum_l [2\la\bR_{,l},\be{l}] - \sum_l [[[\bb,\be{l}],2\la\bR],\be{l}].
$$
Inserting~\eqref{RP} into equation~\eqref{iROm3la}, 
one can express $\tilde{\bB}$ in terms of $\bb$, $\bB$, $\bP$, and $\la$. 
Then the correspondence $\bB\mapsto\tilde{\bB}$ will be a recursion operator 
for symmetries of the Darboux--Egoroff system.  

However, we shall use a slightly different procedure.
Since $\bB$ and $\tilde{\bB}$ are symmetries, $-8\la\tilde{\bB}+4\lambda^2\bB$ 
is a symmetry as well. 
To simplify the formula for the corresponding recursion operator, 
it is more convenient to work 
with the symmetry $-8\la\tilde{\bB}+4\lambda^2\bB$ instead of $\tilde{\bB}$. 

If we denote $\hat{\bB}=-8\la\tilde{\bB}+4\lambda^2\bB$ and insert~\eqref{RP} into~\eqref{iROm3la}, 
the equation~\eqref{iROm3la} becomes 
$$
\numbered\label{hatbbb}
\hat{\bB}\wall 
 = \sum_l [2\la\bR_{,l},\be{l}] - \sum_l [[[\bb,\be{l}],2\la\bR],\be{l}]
\\
= \smash{\sum_{k,l}} (\wall [[\bB_{,kl},\be{k}],\be{l}]
 - [[[[\bb_{,l},\be{k}],\bB],\be{k}],\be{l}]
 - [[[[\bb,\be{k}],\bB_{,l}],\be{k}],\be{l}]
\\ - [[[\bb,\be{l}], [\bB_{,k},\be{k}]],\be{l}]
 + [[[\bb,\be{l}], [[[\bb,\be{k}],\bB],\be{k}]], \be{l}])
\return
- 2 \smash{\sum_{l}} ([[\bb_{,l},\bP],\be{l}]
 + [[\bb,\bP_{,l}],\be{l}]
 - [[[\bb,\be{l}], [\bb,\bP]], \be{l}]
).
\return
$$ 

Using~\eqref{rec}, it is easy to check that for any symmetric 
$n\times n$ matrices $\bM_1$, $\bM_2$ with zero diagonal we have 
$$ 
\sum_i [[\bM_1,\be{i}], [\bM_2,\be{i}]]=-[\bM_1,\bM_2],\qquad\quad
\sum_i [[[\bM_1,\be{i}],\bM_2],\be{i}]=[\bM_1,\bM_2]. 
$$
Using this property, one gets 
$$
\numbered\label{bbbb}
\sum_k
[[[\bb_{,l},\be{k}],\bB],\be{k}]=[\bb_{,l},\bB],\qquad
\sum_k
[[[\bb,\be{k}],\bB_{,l}],\be{k}]=[\bb,\bB_{,l}],\qquad
\sum_k 
[[[\bb,\be{k}],\bB],\be{k}]=[\bb,\bB],\\
\sum_{k,l}
[[[\bb,\be{l}], [\bB_{,k},\be{k}]],\be{l}]=
\sum_{k}
[[\bb,\bB_{,k}],\be{k}]+
\sum_{k,l}
[[\bB_{,k},[[\bb,\be{l}], \be{k}]],\be{l}].
$$
Since $\bP_{,l}=\diag[[\bb,\be{l}],\bB]$ in view of~\eqref{P}, we get
$$
\numbered\label{ppp}
[[\bb,\bP_{,l}],\be{l}]
 =[[\bb,\be{l}],\bP_{,l}]
 =[[\bb,\be{l}],\diag[[\bb,\be{l}],\bB]].
$$
Substituting~\eqref{bbbb} and~\eqref{ppp} in~\eqref{hatbbb}, one obtains
$$ 
\numbered\label{bhdir}
\hat{\bB} = \smash{\sum_{k,l}} (\wall [[\bB_{,kl},\be{k}],\be{l}] %
 + [[[[\bb,\be{l}],\be{k}],\bB_{,k}],\be{l}]) %
\return
\quad
 + \smash{\sum_{k}} (\wall 
[[[\bb,\be{k}], [\bb,\bB]],\be{k}]-[[\bb_{,k},\bB],\be{k}]+2\,[[\bB_{,k}, \bb], \be{k}]\\ %
 - 2\,[[\bb,\be{k}],\diag[[\bb,\be{k}],\bB]] %
 - 2\,[[\bb_{,k}, \bP], \be{k}] 
 + 2\,[[[\bb,\be{k}],[\bb,\bP]],\be{k}]).
\return
$$

By a straightforward computation, one can show the following. 
If $\bB$ is a symmetry and a diagonal matrix $\bP$ obeys  
$\bP_{,k}=\diag[[\bb,\be{k}],\bB]$ for all $k=1,\dots,n$, 
then $\hat{\bB}$ given by~\eqref{bhdir} 
is another symmetry. 
One can also check this property, 
using the component description of~\er{bhdir} given in Proposition~\ref{prop:RO} below. 

We call the correspondence $\bB\mapsto\hat{\bB}$ 
the \emph{direct recursion operator}~$\hat{\mathfrak{R}}$ 
for the Darboux--Egoroff system. 

The same recursion operator can be written in components as follows.
Let $\pot_1,\dots,\pot_n$ be the diagonal elements of the diagonal matrix~$\bP$. 
Then the equation $\bP_{,k}=\diag[[\bb,\be{k}],\bB]$ reads 
$$
\numbered\label{Psi}
\pot_{i,k} =
\begin{cases}
  -2\sum_s \bet_{is} B_{is}, & i = k; \\
  2\bet_{ik} B_{ik}, & i\neq k,
\end{cases}
$$
(some interpretation of~\eqref{Psi} is discussed in Remark~\ref{rempotent} below).

\begin{proposition} 
\label{prop:RO}
Let $\bB=(B_{ij})$ be a symmetry of the Darboux--Egoroff system and $\pot_i$ 
satisfy~\eqref{Psi}. 
Then $\hat{\bB}=(\hat{B}_{ij})$ given by $\hat{B}_{ii}=0$, 
\begin{equation}
\numbered\label{RO}
\hat{B}_{ij} =
B_{ij,ii}
 + \sum_s (\bet_{is} B_{sj,s}
 + \bet_{is}^2 B_{ij}
 + 3 \bet_{ij} \bet_{is} B_{is}
 + (\bet_{is,i} - \bet_{ij} \bet_{sj} + \sum_r \bet_{ir} \bet_{sr}) B_{sj})
\\\qquad  + \bet_{ij,i} \pot_i
 - (\bet_{ij,i} + \sum_s \bet_{is} \bet_{sj}) \pot_j
 + \sum_s \bet_{is} \bet_{sj} \pot_s,\quad\qquad i\neq j,
\end{equation}
is a symmetry of the Darboux--Egoroff system as well. 
The right-hand side of~\er{RO} is equal to the $(i,j)$-th component of the right-hand side 
of~\er{bhdir} divided by~$4$. 
\end{proposition}
\begin{proof}
Analogously to Proposition~\ref{prop:iRO}, 
using equations~\er{DE},~\er{bijjiji},~\er{lDE}, 
one can prove this by a straightforward computation.
\end{proof}

Alternatively, one can obtain the recursion operator~\er{RO} as a reduction 
of the direct recursion operator of the IGSG equation~\cite{M-P},
using the fact that the Darboux--Egoroff system is a reduction of IGSG, 
as has been discussed in Remark~\ref{remdeigsg}.

Let us demonstrate the action of the recursion operator~\er{RO} on some symmetries of the 
Darboux--Egoroff system. 
Recall that $\bet_{ij}$ obey~\er{DE},~\er{bijjiji}. 
Taking $B_{ij} = 0$ as the seed symmetry, 
we see that equations~\eqref{Psi} read $\pot_{i,k}=0$, 
so $\pot_i$ are simply constants, $\pot_i = c_i$, 
and then $\hat{B}_{ij} = \sum_s c_s \bet_{ij,s}$. 
Consequently, the first members of the symmetry hierarchy
are the translations $\bet_{ij,l}$, $l = 1,\dots,n$.

In the next step we take the translation $B_{ij} = \bet_{ij,l}$ as the seed symmetry.
Then 
$$
\pot_i =
\begin{cases}
 -\sum_s \bet_{sl}^2, & i = l; \\
 \bet_{il}^2, & i\neq l,
\end{cases}
$$
satisfy~\eqref{Psi} for $B_{ij} = \bet_{ij,l}$, and formula~\eqref{RO} 
for $i \neq j$ gives
$$
\numbered\label{bijlll}
\hat{B}_{ij} = \bet_{ij,lll} + 3 \times
\begin{cases}
 \sum_s \bet_{ls} (\bet_{lj} \bet_{ls})_{,l}, & i = l; \\
 \sum_s \bet_{ls} (\bet_{il} \bet_{ls})_{,l}, & j = l; \\
 - \bet_{il,l} \bet_{lj,l} + \bet_{il} \bet_{lj} \sum_s \bet_{ls}^2, & i\neq l\neq j\neq i.
\end{cases}
$$

As usual, we can consider the flow associated with the constructed symmetry $\hat{B}_{ij}$. 
Namely, we assume that $\bet_{ij}$ depend on a parameter $t$ and consider the equations 
$\frac{\pd\bet_{ij}}{\pd t} = \hat{B}_{ij}$.

Since the only derivatives in~\eqref{bijlll} are with respect to $x_l$, 
the flow $\frac{\pd\bet_{ij}}{\pd t} = \hat{B}_{ij}$
can be interpreted as a system of evolution equations in the single spatial variable~$x_l$.
This system naturally decomposes into three subsystems:
$i = l$, i.e.,
$$
\numbered\label{el}
\frac{\pd\bet_{lj}}{\pd t} = \bet_{lj,lll} + 3\sum_s \bet_{ls} (\bet_{lj} \bet_{ls})_{,l},
$$
and similarly for $j = l$,
and the remaining equations 
$$
\numbered\label{er}
\frac{\pd\bet_{ij}}{\pd t} = 
\bet_{ij,lll} - 3 \bet_{li,l} \bet_{lj,l} + 3 \bet_{li} \bet_{lj} \sum_s \bet_{ls}^2,\qquad\quad
i\neq l\neq j\neq i.
$$
Denoting $x = x_l$, $u^j = \bet_{lj}$, we see that the subsystem~\eqref{el} becomes
$$
\numbered\label{elu}
u^j_t = u^j_{xxx} + 3 \sum_s u^s (u^j u^s)_x.
$$
Omitting the zero component $u^l = \bet_{ll} = 0$, let us combine the remaining components 
$u^j$, $j \ne l$, into an $(n-1)$-component vector $u$. 
Then~\eqref{elu} can be compactly written as
$$
\numbered\label{vmKdV}
u_t = u_{xxx} + 3 (u,u) u_x + 3 (u,u_x) u,
$$
which is the integrable 
{\it vector modified Korteweg--de Vries equation}~\cite[eq.~(3.9), $m = 1$]{A-F}, 
henceforth abbreviated as vmKdV. 
For its B\"acklund transformation and recursion operator, see~\cite[eq.~(4.10)]{B-M} 
and~\cite[eq.~(47)]{An}, respectively.
For a broader view, see~\cite[eq.~(11)]{S-S} or~\cite[eq.~(1.4)]{O-S}.

Observe that the remaining equations~\eqref{er} are linear in the remaining 
unknowns $\bet_{ij}$, $i,j \neq l$. One can say that the nonlinear substance of the 
symmetry~\eqref{bijlll} lies in the vmKdV equation~\eqref{vmKdV}.

Let us apply the recursion operator from Proposition~\ref{prop:RO} 
to the symmetry~\eqref{bijlll}. 
Let $B_{ij}$ for $i\neq j$ be the right-hand side of~\eqref{bijlll}.  
Then 
$$
\pot_i =
\begin{cases}
- 2 \sum_s \beta_{sl,ll} \beta_{sl} + \sum_s \beta_{sl,l}^2 - 3 \sum_s (\beta_{sl}^2)^2, 
 & i = l; \\
2 \beta_{il,ll} \beta_{il} - \beta_{il,l}^2 + 3 \beta_{il}^2 \sum_s \beta_{ls}^2, 
& i\neq l,
\end{cases}
$$
satisfy~\eqref{Psi}, and formula~\eqref{RO} gives a symmetry of order $5$. 
Similarly to the above discussion, 
in the structure of this symmetry one finds the 5th order vmKdV equation
$$
u_t = u_{xxxxx}
 + 5 (u,u) u_{xxx}
 + 15 (u,u_x) u_{xx}
 + 15 (u,u_{xx}) u_x
 + 10 (u_x,u_x) u_x
\\\quad
 + 10 (u,u)^2 u_x
 + 5 (u,u_{xxx}) u
 + 10 (u_x,u_{xx}) u
 + 20 (u,u)(u,u_x) u.
$$
The 5th order vmKdV equation is given also in~\cite[eq.~(62)]{An}, 
but one of the coefficients is different in~\cite[eq.~(62)]{An}, 
which seems to be a misprint.

\begin{remark}
\lb{rempotent}
Note that~\eqref{Psi} can be interpreted as follows. 
The Darboux--Egoroff system~\eqref{DE} possesses potentials $p_1,\dots,p_n$ 
defined by the compatible equations
$$
\numbered\label{psi}
p_{i,k} =
\begin{cases}
  -\sum_s \bet_{is}^2, & i = k; \\
  \bet_{ik}^2, & i\neq k.
\end{cases}
$$
Since the right-hand side of~\eqref{Psi} is obtained 
from the right-hand side of~\eqref{psi} by applying the linearization operator $\lin_{\bB}$, 
one can say that~\eqref{Psi} is the linearization of~\eqref{psi}. 
\end{remark}

Note that 
the Darboux--Egoroff system has the scaling symmetry $\bB^{*}=\bb+\sum_r x_r\bb_{,r}$. 
For this symmetry, equations~\er{Psi} read
$$
\numbered\label{Psisc}
\pot_{i,k} =
\begin{cases}
  -2\sum_{s} \bet_{is} (\bet_{is}+\sum_r x_r\bet_{is,r}), & i = k; \\
  2\bet_{ik} (\bet_{ik}+\sum_r x_r\bet_{ik,r}), & i\neq k.
\end{cases}
$$
Using the potentials $p_i$ defined by~\er{psi}, 
it is easy to check that $\pot_i = p_i + \sum_s (x_s - x_i) \bet_{is}^2$ satisfy~\er{Psisc}.

Applying the recursion operator~\eqref{RO} to $\bB^{*}$, one obtains 
a nonlocal symmetry depending on $x_k$, $\bet_{ij}$, $\bet_{ij,k}$, $\bet_{ij,kl}$, 
$\bet_{ij,klm}$, and $\pot_i = p_i + \sum_s (x_s - x_i) \beta_{is}^2$. 
(More precisely, 
one obtains a symmetry shadow in the covering determined by the compatible system~\eqref{psi}, 
but, in the literature on symmetries of PDEs, 
symmetry shadows of this type are usually called nonlocal symmetries.) 
The repeated application of the recursion operator to $\bB^{*}$  
yields an infinite number of symmetry shadows for the Darboux--Egoroff system. 

\section{An infinite-dimensional Lie algebra of nonlocal symmetries}
\lb{secnonlsym}

Recall that, in matrix notation, the Darboux--Egoroff system~\er{DEm} reads
\begin{equation}
\label{demy}
[\bb,\be{k}]_{,l} - [\bb,\be{l}]_{,k} + [[\bb,\be{k}], [\bb,\be{l}]] = 0,\qquad\quad 
k,l=1,\dots,n,
\end{equation}
and possesses the ZCR 
\beq
\lb{zcrdemy}
\bA^{k}_{,l}-\bA^{l}_{,k}+[{\bA}^{k},{\bA}^{l}]=0,\qquad \bA^{k}=\lambda\be{k}+[\bb,\be{k}], 
\qquad k,l=1,\dots,n.
\ee

The corresponding auxiliary linear system
\begin{equation}
\label{laps}
\Psi_{,k} = \bA^{k} \Psi=(\lambda \be{k} + [\bb,\be{k}])\Psi, \qquad \det \Psi \neq 0,
\qquad k=1,\ldots,n,
\end{equation}
is compatible modulo~\eqref{demy}. Here $\Psi$ is an invertible $n\times n$ matrix-function. 
According to the terminology described in Section~\ref{sec-prelim}, 
the compatible system~\eqref{laps} determines a covering of~\eqref{demy} 
with the matrix pseudopotential $\Psi$.

As has been said in Section~\ref{convnot}, 
the entries of all considered matrices belong to~$\fik$, where $\fik$ is either $\Com$ or $\mathbb{R}$. 
For any $\sm\in\zsp$, we denote by $\gl_{\sm}(\fik)$ the algebra of all $\sm\times\sm$ matrices 
and by $\mathrm{GL}_{\sm}(\fik)$ the group of invertible $\sm\times\sm$ matrices 
with entries from $\fik$.

According to~\eqref{bhmain}, for every constant matrix $\bc\in\gl_n(\fik)$, 
the matrix-function 
\beq
\lb{shadts}
\hat{\bB}=\nd(\Psi\bc \Psi^{-1} + (\Psi\bc\Psi^{-1})^\top)
\ee 
is a symmetry shadow for~\eqref{demy} in the covering determined by~\eqref{laps}. 
This symmetry shadow has been obtained in Section~\ref{secinvrec} 
by applying the inverse recursion operator to the zero symmetry 
in the framework of Proposition~\ref{prop:iROtriv}.

Since system~\eqref{laps} depends on the parameter $\la$,
we can use the standard idea of introducing 
a formal Taylor series expansion $\Psi=\sum_{i=0}^\infty\la^i \Psi_i$ 
for~$\Psi$ in~$\lambda$, where $\det \Psi_0\neq 0$. 
Substituting $\Psi=\sum_{i=0}^\infty\la^i \Psi_i$ in equation~\eqref{laps}, one gets 
\begin{equation}
\label{cw0ik}
\Psi_{0,k} =[\bb,\be{k}]\Psi_0,\qquad 
\Psi_{j,k} =\be{k}\Psi_{j-1} + [\bb,\be{k}]\Psi_j,\qquad
j\in\zsp,\qquad k=1,\dots,n. 
\end{equation}
As system~\eqref{cw0ik} is compatible modulo~\eqref{demy}, 
we can regard $\Psi_i$ as matrix pseudopotentials for~\eqref{demy}. 

Denote by $\mathrm{GL}_n(\fik[[\la]])$ the group of formal power series of the form 
\beq
\lb{lacwcw0}
\sum_{i=0}^\infty\la^i \Psi_i,\qquad\Psi_i\in\gl_n(\fik),\qquad
\det \Psi_0\neq 0.
\ee

Substituting $\Psi=\sum_{i=0}^\infty\la^i \Psi_i$ in the symmetry shadow~\eqref{shadts} 
and considering the terms of degree zero in~$\la$, we see that the matrix-function 
\beq
\lb{shadcc}
\bB'=\nd(\Psi_0\bc \Psi_0^{-1} + (\Psi_0\bc\Psi_0^{-1})^\top)
\ee 
is a symmetry shadow in the covering determined by~\eqref{cw0ik}. 

Let $\gl_n(\fik[[\la]])$ be the algebra of formal power series 
$\sum_{i=0}^\infty\la^i P_i$ with coefficients $P_i\in\gl_n(\fik)$.
For any 
$$
P(\la)=\sum_{i=0}^\infty\la^i P_i\,\in\,\gl_n(\fik[[\la]]),
$$ 
set $P^\top(-\la)=\sum_{i=0}^\infty(-\la)^i P_i^\top$.

Since the functions $\bA^{k} = \lambda \be{k} + [\bb,\be{k}]$ 
from \eqref{zcrdemy} are polynomial in $\la$, 
we can regard $\bA^{k}$ as $\gl_n(\fik[[\la]])$-valued functions. 
Since the matrix $\be{k}$ is symmetric and $[\bb,\be{k}]$ is skew-symmetric, 
the functions $\bA^{k}$ take values in the Lie subalgebra 
$$
\tpa=\left.\{P(\la)\in\gl_n(\fik[[\la]])\,\ \ \right|\,\ \ P(\la)+P^\top(-\la)=0\}.
$$
One has $P(\la)=\sum_{i=0}^\infty\la^i P_i\,\in\,\tpa$ 
if and only if $P_i^\top=(-1)^{i+1}P_i$ for all $i$.

\begin{remark}
\lb{remobserv} 
The following simple well-known observation will be useful for us. 

Let $\sm\in\zsp$. 
Let $G\subset\mathrm{GL}_{\sm}(\fik)$ be a closed Lie subgroup of $\mathrm{GL}_{\sm}(\fik)$ 
and $\mg\subset\gl_{\sm}(\fik)$ be the Lie algebra of~$G$. 
Suppose that $\mg$-valued functions $\fzc^{k}=\fzc^{k}(x_1,\dots,x_n)$, $k=1,\ldots,n$, 
satisfy 
$$
\fzc^{k}_{,l} - \fzc^{l}_{,k} + [{\fzc}^{k},{\fzc}^{l}] = 0,\qquad 
k,l=1,\ldots,n.
$$
Then the following system is compatible 
\beq
\lb{covglsmxx}
Q_{,k}=\fzc^{k}Q,\qquad k=1,\ldots,n,\qquad Q\in\mathrm{GL}_{\sm}(\fik),
\ee
where $Q=Q(x_1,\dots,x_n)$ is a $\mathrm{GL}_{\sm}(\fik)$-valued function.  
As $\fzc^{k}$ take values in the Lie algebra~$\mg$ corresponding to the Lie subgroup 
$G\subset\mathrm{GL}_{\sm}(\fik)$, the following system is compatible as well 
\beq
\lb{covglsmxxgg}
\tilde{Q}_{,k}=\fzc^{k}\tilde{Q},\qquad k=1,\ldots,n,\qquad \tilde{Q}\in G,
\ee
where $\tilde{Q}=\tilde{Q}(x_1,\dots,x_n)$ takes values in~$G$.
We say that~\eqref{covglsmxxgg} is the reduction of system~\eqref{covglsmxx} to the subgroup 
$G\subset\mathrm{GL}_{\sm}(\fik)$.
\end{remark}

Let us try to use the idea of Remark~\ref{remobserv} in the case when 
$\mathrm{GL}_{\sm}(\fik)$ is replaced by $\mathrm{GL}_n(\fik[[\la]])$.
Recall that system~\eqref{cw0ik} can be written as 
\beq
\lb{lsdela}
\Psi(\la)_{,k} = \bA^{k}\Psi(\la)=(\lambda \be{k}+[\bb,\be{k}])\Psi(\la),\qquad k=1,\ldots,n, 
\qquad \Psi(\la)=\sum_{i=0}^\infty\la^i \Psi_i\in\mathrm{GL}_n(\fik[[\la]]).
\ee
We know that the functions $\bA^{k}=\lambda \be{k} + [\bb,\be{k}]$ 
take values in the Lie subalgebra $\tpa\subset\gl_n(\fik[[\la]])$.
In order to make for system~\eqref{lsdela} an analog of the reduction described in Remark~\ref{remobserv}, 
we need to define a ``Lie subgroup of $\mathrm{GL}_n(\fik[[\la]])$ corresponding 
to the Lie subalgebra $\tpa\subset\gl_n(\fik[[\la]])$''.

However, since $\dim\tpa=\infty$, it is not obvious how to define such a Lie subgroup. 
Nevertheless, one can use the following reasoning. 
If a ``Lie subgroup of $\mathrm{GL}_n(\fik[[\la]])$ corresponding 
to the Lie subalgebra $\tpa\subset\gl_n(\fik[[\la]])$'' can be defined 
in any reasonable sense, then this subgroup must contain the elements $\exp(P(\la))$ for 
all $P(\la)\in\tpa$. 

Each $P(\la)\in\tpa$ satisfies $P(\la)+P^\top(-\la)=0$, which implies that 
$R(\la)=\exp(P(\la))$ obeys 
$R(\la)\cdot R^\top(-\la)=\bE$,
where $\bE$ is the unit matrix. 
Therefore, it makes sense to consider the \mbox{subgroup} 
\beq
\lb{tgrdef}
\tgr=\left.\{\,\Psi(\la)\in\mathrm{GL}_n(\fik[[\la]])\ \ \ \right|
\ \ \ \Psi(\la)\cdot\Psi^\top(-\la)=\bE\,\}.
\ee

So let us impose the constraint 
\beq
\lb{psilabe}
\Psi(\la)\cdot\Psi^\top(-\la)=\bE,\qquad\quad \Psi(\la)=\sum_{i=0}^\infty\la^i \Psi_i.
\ee
It is easy to check that~\eqref{psilabe} is compatible with~\eqref{lsdela}. 
Informally speaking, by analogy with Remark~\ref{remobserv}, 
one can say that system~\eqref{lsdela},~\eqref{psilabe}  
is the reduction of system~\eqref{lsdela} to the subgroup $\tgr\subset\mathrm{GL}_n(\fik[[\la]])$.

However, we do not consider any Lie group structure on~\eqref{tgrdef}. 
We need only the fact that the constraint~\eqref{psilabe} is compatible with equations~\eqref{lsdela}.
The above discussion on Lie groups is just a motivation for introducing the constraint~\eqref{psilabe}.

The idea to use the subgroup~\eqref{tgrdef} in connection with 
the Darboux--Egoroff system~\eqref{demy} is known~(see, e.g.,~\cite{vdltw}). 
Our aim in the above discussion has been to explain how one can guess that 
this subgroup should be considered, 
by analogy with the situation discussed in Remark~\ref{remobserv}.

Recall that~\eqref{shadcc} is a symmetry shadow. 
Equation~\eqref{psilabe} implies $\Psi_0\Psi_0^\top=\bE$, hence $\Psi_0^{-1}=\Psi_0^\top$. 
Therefore, after imposing the constraint~\eqref{psilabe}, we can rewrite formula~\eqref{shadcc} as 
\beq
\lb{hatbh}
\bB'=\nd(\Psi_0(\bc+\bc^\top)\Psi_0^\top).
\ee
To obtain~\eqref{hatbh}, we have substituted $\Psi_0^{-1}=\Psi_0^\top$ in~\eqref{shadcc}. 

Let $\br=\bc+\bc^\top$. Then $\br^\top=\br$. 
Since,  for any constant matrix $\bc\in\gl_n(\fik)$, 
the function~\eqref{hatbh} is a symmetry shadow in the covering determined by 
the compatible system~\eqref{lsdela},~\eqref{psilabe}, we see that 
\beq
\lb{bhbr}
\bB_{\br}=\nd\Psi_0\br\Psi_0^\top
\ee
is a symmetry shadow in this covering  
for any constant matrix $\br\in\gl_n(\fik)$ satisfying $\br^\top=\br$. 
Let us find a nonlocal symmetry such that the corresponding symmetry shadow is~\eqref{bhbr}.

According to the definition of nonlocal symmetries presented in Section~\ref{sec-prelim}, 
we need to find matrix-functions 
$$
\nsy_i=\nsy_i(\Psi_l,x_k,\bb,\dots),\qquad\quad i\in\zp,
$$ 
so that the following property holds. 
If $\bb$, $\Psi_i$ satisfy~\eqref{demy},~\eqref{lsdela},~\eqref{psilabe}  
then the infinitesimal deformation 
\beq
\lb{deformnsy}
\tilde{\bb}=\bb+\ve\nd\Psi_0\br\Psi_0^\top,\qquad\quad
\tilde{\Psi}_i=\Psi_i+\ve\nsy_i, \qquad\quad i\in\zp,
\ee
obeys equations~\eqref{demy},~\eqref{lsdela},~\eqref{psilabe} up to $O(\varepsilon^2)$. 
Here $\nsy_i$ may depend on $\Psi_l$, $x_k$, $\bb$, and derivatives of $\bb$. 

Considering the terms of degree zero in~$\la$ in equations~\eqref{lsdela}, we get 
\beq
\lb{psi0kbh}
\Psi_{0,k} =[\bb,\be{k}]\Psi_0,\qquad\quad k=1,\dots,n.
\ee
Substituting $\tilde{\bb}=\bb+\ve\nd\Psi_0\br\Psi_0^\top$ and $\tilde{\Psi}_0=\Psi_0+\ve\nsy_0$ 
in place of $\bb$, $\Psi_0$ in~\eqref{psi0kbh}   
and collecting the terms linear in~$\ve$, one obtains 
\beq
\lb{nsy0k}
\nsy_{0,k} =[\nd\Psi_0\br\Psi_0^\top,\be{k}]\Psi_0+[\bb,\be{k}]\nsy_0,\qquad\quad k=1,\ldots,n. 
\ee 

Let us first try to find a function $\nsy_0$ such that $\nsy_0$ obeys~\eqref{nsy0k}. 
From~\eqref{lsdela},~\eqref{psi0kbh} one gets
\beq
\lb{pdpsipsi}
\Psi_{0,k} =[\bb,\be{k}]\Psi_0,\qquad\quad
\Psi^\top_{0,k}=-\Psi^\top_0[\bb,\be{k}],\qquad\quad
\Psi_{1,k} =\be{k}\Psi_{0} + [\bb,\be{k}]\Psi_1.
\ee 
The form of equations~\eqref{pdpsipsi},~\eqref{nsy0k}
suggests to seek a function $\nsy_0$ equal to a polynomial in 
$\br$, $\Psi_0$, $\Psi_0^{\top}$, $\Psi_{1}$. 
Trying low-degree polynomials, one finds that 
\beq
\lb{nsy0for}
\nsy_0=\Psi_0\br\Psi_0^{\top}\Psi_{1}-\Psi_{1}\br
\ee 
is suitable. 
Indeed, using~\eqref{pdpsipsi},~\eqref{nsy0for}, $\Psi_0^\top\Psi_0=\bE$, 
and $[\nd\Psi_0\br\Psi_0^\top,\be{k}]=[\Psi_0\br\Psi_0^\top,\be{k}]$, we obtain 
\beq
\lb{pdnsy0f}
\nsy_{0,k}\wall=
\Psi_{0,k}\br\Psi_0^{\top}\Psi_{1}
+\Psi_0\br\Psi_{0,k}^{\top}\Psi_{1}
+\Psi_0\br\Psi_{0}^{\top}\Psi_{1,k}-\Psi_{1,k}\br\\
=[\bb,\be{k}]\Psi_0\br\Psi_0^\top\Psi_1
-\Psi_0\br\Psi_0^\top[\bb,\be{k}]\Psi_1
+\Psi_0\br\Psi_0^\top(\be{k}\Psi_{0} + [\bb,\be{k}]\Psi_1)-(\be{k}\Psi_{0} + [\bb,\be{k}]\Psi_1)\br\\
=[\bb,\be{k}]\Psi_0\br\Psi_0^\top\Psi_1
+\Psi_0\br\Psi_0^\top\be{k}\Psi_{0}
-(\be{k}\Psi_{0} + [\bb,\be{k}]\Psi_1)\br,
\return
\ee
\beq
\lb{psinsy0}
[\nd\Psi_0\br\Psi_0^\top,\be{k}]\Psi_0+[\bb,\be{k}]\nsy_0=
[\Psi_0\br\Psi_0^\top,\be{k}]\Psi_0+[\bb,\be{k}]\nsy_0\\
\qquad=\Psi_0\br\Psi_0^\top\be{k}\Psi_0-\be{k}\Psi_0\br\Psi_0^\top\Psi_0
+[\bb,\be{k}]\nsy_0
=\Psi_0\br\Psi_0^\top\be{k}\Psi_0-\be{k}\Psi_0\br
+[\bb,\be{k}](\Psi_0\br\Psi_0^{\top}\Psi_{1}-\Psi_{1}\br).
\ee
Comparison of~\eqref{pdnsy0f} and~\eqref{psinsy0} shows  
that $\nsy_0=\Psi_0\br\Psi_0^{\top}\Psi_{1}-\Psi_{1}\br$ 
obeys~\eqref{nsy0k}.

Now~\eqref{nsy0for} suggests to try the following formulas for functions $\nsy_i$
\beq
\lb{nsyif}
\nsy_i=\Psi_0\br\Psi_0^{\top}\Psi_{i+1}-\Psi_{i+1}\br,\qquad\quad i\in\zp.
\ee 
The next proposition shows that~\eqref{nsyif} satisfy the required properties. 
\begin{proposition}
\lb{propns}
For any constant symmetric matrix $\br\in\gl_n(\fik)$, 
consider $\nsy_i$ given by~\eqref{nsyif}.
Then the infinitesimal deformation~\eqref{deformnsy} 
is a symmetry of system~\eqref{demy},~\eqref{lsdela},~\eqref{psilabe}.
We denote this symmetry by $S_{\br}$.

Since $\Psi_i$ satisfying~\eqref{lsdela},~\eqref{psilabe} are pseudopotentials 
for the Darboux--Egoroff system~\eqref{demy}, 
the symmetry $S_{\br}$ is a nonlocal symmetry of the Darboux--Egoroff system~\eqref{demy} 
in the covering determined by~\eqref{lsdela},~\eqref{psilabe}.

\end{proposition}
\begin{proof}
Let $\bb$ and $\Psi(\la)=\sum_{i=0}^\infty\la^i \Psi_i$ satisfy~\eqref{demy},~\eqref{lsdela},~\eqref{psilabe}.  
Set $\nsy(\la)=\sum_{i=0}^\infty\la^i\nsy_i$. 
Using~\eqref{nsyif} and $\Psi_0^\top\Psi_0=\bE$, one obtains 
$$
\nsy(\la)=\sum_{i=0}^\infty\la^i\nsy_i=
\sum_{i=0}^\infty\la^i(\Psi_0\br\Psi_0^{\top}\Psi_{i+1}-\Psi_{i+1}\br)=
\la^{-1}(\Psi_0\br\Psi_0^{\top}\Psi(\la)-\Psi(\la)\br).
$$

To prove that the infinitesimal deformation~\eqref{deformnsy} 
is a symmetry of system~\eqref{demy},~\eqref{lsdela},~\eqref{psilabe},  
we need to show that $\tilde{\bb}=\bb+\ve\nd\Psi_0\br\Psi_0^\top$ 
and $\tilde{\Psi}(\la)=\Psi(\la)+\ve\nsy(\la)$ obey 
equations~\eqref{demy},~\eqref{lsdela},~\eqref{psilabe} up to $O(\varepsilon^2)$. 

Since~\eqref{bhbr} is a symmetry shadow, $\tilde{\bb}=\bb+\ve\nd\Psi_0\br\Psi_0^\top$ satisfies  
equations~\eqref{demy} up to $O(\varepsilon^2)$. 

Using~\eqref{psilabe} and $\nsy(\la)=\la^{-1}(\Psi_0\br\Psi_0^{\top}\Psi(\la)-\Psi(\la)\br)$, one gets
\beq
\lb{tplatp}
\tilde{\Psi}(\la)\tilde{\Psi}^\top(-\la) =
(\Psi(\la)+\ve\nsy(\la))(\Psi^\top(-\la)+\ve\nsy^\top(-\la))
=\bE+\ve\big(\Psi(\la)\nsy^\top(-\la)+\nsy(\la)\Psi^\top(-\la)\big)+
\ve^2\nsy(\la)\nsy^\top(-\la), 
\ee
\beq
\lb{planst}
\Psi(\la)\nsy^\top(-\la)+\nsy(\la)\Psi^\top(-\la)\wall =
\Psi(\la)(-\la^{-1}\left(\Psi^\top(-\la)\Psi_0\br\Psi_0^{\top}-\br\Psi^\top(-\la)\right))
+\la^{-1}(\Psi_0\br\Psi_0^{\top}\Psi(\la)-\Psi(\la)\br)\Psi^\top(-\la)
\return
\\
\qqquad\qqquad=-\la^{-1}(\Psi_0\br\Psi_0^{\top}-\Psi(\la)\br\Psi^\top(-\la))
+\la^{-1}(\Psi_0\br\Psi_0^{\top}-\Psi(\la)\br\Psi^\top(-\la))=0.
\ee
Equations~\eqref{tplatp},~\eqref{planst} yield  
$\tilde{\Psi}(\la)\tilde{\Psi}^\top(-\la)=\bE+\ve^2\nsy(\la)\nsy^\top(-\la)$, 
so $\tilde{\Psi}(\la)$ obeys~\eqref{psilabe} up to $O(\varepsilon^2)$.

Using~\eqref{lsdela},~\eqref{pdpsipsi}, 
$\nsy(\la)=\la^{-1}(\Psi_0\br\Psi_0^{\top}\Psi(\la)-\Psi(\la)\br)$, and 
$[\nd\Psi_0\br\Psi_0^\top,\be{k}]=[\Psi_0\br\Psi_0^\top,\be{k}]$,  
for $k=1,\dots,n$ we obtain 
$$
\tilde{\Psi}(\la)_{,k}=\Psi(\la)_{,k}+\ve\nsy(\la)_{,k}=
\Psi(\la)_{,k}+\ve(\la^{-1}(\Psi_0\br\Psi_0^{\top}\Psi(\la)-\Psi(\la)\br))_{,k}\\
\quad\wall=\Psi(\la)_{,k}+\ve\la^{-1}
([\bb,\be{k}]\Psi_0\br\Psi_0^{\top}\Psi(\la)
-\Psi_0\br\Psi^\top_0[\bb,\be{k}]\Psi(\la)
+\Psi_0\br\Psi_0^{\top}(\lambda \be{k}+[\bb,\be{k}])\Psi(\la)
-(\lambda \be{k}+[\bb,\be{k}])\Psi(\la)\br)\\
=(\lambda \be{k}+[\bb,\be{k}])\Psi(\la)+\ve
((\lambda \be{k}+[\bb,\be{k}])\la^{-1}(\Psi_0\br\Psi_0^{\top}\Psi(\la)-\Psi(\la)\br)
+[\Psi_0\br\Psi^\top_0,\be{k}]\Psi(\la))\\
=(\lambda \be{k}+[\bb+\ve\Psi_0\br\Psi^\top_0,\be{k}])(\Psi(\la)+\ve\nsy(\la))
-\ve^2[\Psi_0\br\Psi^\top_0,\be{k}]\nsy(\la)\\
=(\lambda \be{k}+[\tilde{\bb},\be{k}])\tilde{\Psi}(\la)-\ve^2[\Psi_0\br\Psi^\top_0,\be{k}]\nsy(\la),
\return
$$
which implies that $\tilde{\bb}$, $\tilde{\Psi}(\la)$ obey~\eqref{lsdela} up to $O(\varepsilon^2)$. 
\end{proof}

According to Proposition~\ref{propns}, 
for each $\br\in\gl_n(\fik)$ satisfying $\br^\top=\br$, 
we have the symmetry~$S_{\br}$ of system~\eqref{demy},~\eqref{lsdela},~\eqref{psilabe}, 
where $S_{\br}$ is given by the infinitesimal deformation~\eqref{deformnsy} 
with $\nsy_i=\Psi_0\br\Psi_0^{\top}\Psi_{i+1}-\Psi_{i+1}\br$.

Let $\sla$ be the Lie algebra generated by the symmetries $S_{\br}$ 
for all symmetric matrices $\br\in\gl_n(\fik)$.
Then every element of~$\sla$ is a symmetry of system~\eqref{demy},~\eqref{lsdela},~\eqref{psilabe}.
Therefore, every element of~$\sla$ can be regarded as a nonlocal symmetry 
of the Darboux--Egoroff system~\eqref{demy} 
in the covering determined by~\eqref{lsdela},~\eqref{psilabe}. 

The explicit structure of $\sla$ is described in Proposition~\ref{propalgs} below.
To present this description, we need to consider some infinite-dimensional Lie algebras. 
As has been discussed in Remark~\ref{remburshad} in Section~\ref{sec-intr}, 
we denote by $\twl$ the Lie algebra of Laurent polynomials of the form 
$$
\sum_{i=-p}^q\la^i\br_i,\qquad p,q\in\zp,\qquad \br_i\in\gl_n(\fik),\qquad
\br_i^\top=(-1)^{i+1}\br_i.
$$
The algebra $\twl$, as a vector space, is equal to the direct sum of the subalgebras 
$$
\twl_{+}=\left.\{\sum_{i=0}^q\la^i\br_i\ \ \right|\ \ q\ge 0,\ \ \ \br_i^\top=(-1)^{i+1}\br_i\big\},
\qquad\quad
\twl_{-}=\left.\big\{\sum_{i=-p}^{-1}\la^i\br_i\ \ \right|\ \ p>0,\ \ \ \br_i^\top=(-1)^{i+1}\br_i\big\}.
$$

\begin{proposition}
\lb{propalgs} 
The algebra $\sla$ is isomorphic to the Lie algebra 
$$
\twl'_{-}=
\left.\big\{\,\sum_{i=-p}^{-1}\la^i\br_i\,\in\,\twl_{-}
\ \ \ \ \right|\ \ \ \ 
p>0,\quad\br_i\in\gl_n(\fik),\quad\br_i^\top=(-1)^{i+1}\br_i,\quad
\tr(\br_i)=0\big\,\},
$$
where $\tr(\br_i)$ is the trace of $\br_i$. 
Under this isomorphism, $S_{\br}\in\sla$ corresponds to 
$$
-\la^{-1}\big(\br-\frac1n\tr(\br)\bE\big)\in\twl'_{-}
$$ 
for each $\br\in\gl_n(\fik)$ satisfying $\br^\top=\br$.
\end{proposition} 
\begin{proof}
Using the infinitesimal version of the Givental--van de Leur twisted loop group action 
on the space of semisimple Frobenius manifolds~\cite{fvdls,givental1,givental2,lee1,lee2,vdltw}, 
Buryak and Shadrin~\cite{bur-shadr} described an action of the Lie algebra~$\twl$ 
by symmetries of system~\eqref{demy},~\eqref{lsdela},~\eqref{psilabe}. 
In other words, the paper~\cite{bur-shadr} presents a homomorphism from $\twl$ 
to the Lie algebra of symmetries of system~\eqref{demy},~\eqref{lsdela},~\eqref{psilabe}. 
We denote this homomorphism by~$\vf$.

The algebra $\twl$ is spanned by elements of the form $\la^{-\ell}\br\in\twl$, 
where $\ell\in\mathbb{Z}$ and $\br\in\gl_n(\fik)$ are such that $\br^\top=(-1)^{\ell+1}\br$.
Therefore, in order to describe the homomorphism $\vf$, it is sufficient 
to describe $\vf(\la^{-\ell}\br)$. 

According to the formulas from~\cite[Section 3.2]{bur-shadr}, 
the symmetry $\vf(\la^{-\ell}\br)$ is given by the following infinitesimal deformation 
of solutions of system~\eqref{demy},~\eqref{lsdela},~\eqref{psilabe} 
\beq
\lb{vevf}
\tilde{\bb}=\bb+\ve\vf(\la^{-\ell}\br)(\bb),\qquad\quad
\tilde{\Psi}_i=\Psi_i+\ve\vf(\la^{-\ell}\br)(\Psi_i), \qquad\quad i\in\zp,
\ee
\begin{equation}
\lb{vflfbh}
\vf(\la^{-\ell}\br)(\bb)=
\begin{cases}
  \nd\sum_{i,j\ge 0,\, i+j=\ell-1} (-1)^{j-1} \Psi_i\br\Psi_j^{\top}, & \ell>0,\\
  0, & \ell\le 0,
\end{cases}
\end{equation}
\begin{equation}
\lb{vflfpsii}
\vf(\la^{-\ell}\br)(\Psi_i)=
\begin{cases}
\Psi_{\ell+i}\br  
-\sum_{p=1}^\ell\sum_{q=0}^{\ell-p}(-1)^{\ell-p-q}\Psi_q\br \Psi_{\ell-p-q}^{\top}
\Psi_{p+i}, &\ell>0,\\
  \Psi_{\ell+i}\br, & \ell\le 0,\quad i\ge -\ell, \\
  0, & \ell\le 0,\quad i< -\ell.
\end{cases}
\end{equation}

That is, if $\bet$ and $\Psi_i$ satisfy~\eqref{demy},~\eqref{lsdela},~\eqref{psilabe}, 
then $\tilde{\bb}$ and $\tilde{\Psi}_i$ given by~\eqref{vevf} obey 
equations~\eqref{demy},~\eqref{lsdela},~\eqref{psilabe} up to $O(\varepsilon^2)$, 
where $\vf(\la^{-\ell}\br)(\bb)$ and $\vf(\la^{-\ell}\br)(\Psi_i)$ 
are defined by~\eqref{vflfbh},~\eqref{vflfpsii}.

In particular, according to~\eqref{vflfbh},~\eqref{vflfpsii}, 
for $\ell=1$ and $\br\in\gl_n(\fik)$ satisfying $\br^\top=\br$ we have 
$$
\vf(\la^{-1}\br)(\bb)=-\nd\Psi_0\br\Psi_0^{\top},\qquad\quad 
\vf(\la^{-1}\br)(\Psi_i)=\Psi_{i+1}\br-\Psi_0\br\Psi_0^{\top}\Psi_{i+1},
$$
which implies that $\vf(\la^{-1}\br)=-S_{\br}$.  

Let $\cen\subset\twl_{-}$ be the subalgebra spanned by the elements 
$\la^{-(2s+1)}\bE\in\twl_{-}$ for all $s\in\zp$. 
Clearly, $\cen$ is the center of the Lie algebra $\twl_{-}$, and we have $\twl_{-}=\twl'_{-}\oplus\cen$.

Since $\Psi(\la)=\sum_{i=0}^\infty\la^i \Psi_i$ and 
$\Psi^\top(-\la)=\sum_{i=0}^\infty(-\la)^i \Psi_i^\top$, 
equation~\eqref{psilabe} says that 
\beq
\lb{psiijk}
\Psi_0\Psi_0^\top=\bE,\qquad 
\sum_{i,j\ge 0,\, i+j=k}(-1)^j\Psi_i\Psi_j^\top=0\qquad\forall\,k\in\zsp.
\ee
Using~\eqref{vflfbh},~\eqref{vflfpsii}, and~\eqref{psiijk}, 
for any $s\in\zp$ we obtain $\vf(\la^{-(2s+1)}\bE)=0$.
Hence $\vf(\cen)=0$.

It is easily seen that the algebra $\twl'_{-}$ is generated 
by the elements 
$$
\la^{-1}\big(\br-\frac1n\tr(\br)\bE\big)\,\in\,\twl'_{-}
$$
for $\br\in\gl_n(\fik)$ satisfying $\br^\top=\br$.
  
As $\vf(\la^{-1}\br)=-S_{\br}$ and $\vf(\la^{-1}\bE)=0$, we get 
$\vf\big(\la^{-1}(\br-\frac1n\tr(\br)\bE)\big)=-S_{\br}$. Hence $\vf(\twl'_{-})=\sla$.

Recall that the Lie algebra $\twl_{+}$ is spanned by elements of the form 
$\la^{m}\br_m$, where $m\in\zp$ and $\br_m\in\gl_n(\fik)$ are such that $\br_m^\top=(-1)^{m+1}\br_m$.
According to~\eqref{vflfbh},~\eqref{vflfpsii}, for $m\ge 0$ one has 
\beq
\lb{vflam}
\vf(\la^{m}\br_m)(\bb)=0,\qquad\quad 
\vf(\la^{m}\br_m)(\Psi_i)=
\begin{cases}
  \Psi_{i-m}\br_m, & \quad i\ge m, \\
  0, & \quad i< m.
\end{cases}
\ee
Using~\eqref{vflam}, it is easy to show that $\vf(v)\neq 0$ for any nonzero element $v\in\twl_{+}$.
That is, $\twl_{+}\cap\ker\vf=0$.

For any nonzero element $w\in\twl'_{-}$, 
there is $v\in\twl_{+}$ such that $[v,w]\neq 0$ and $[v,w]\in\twl_{+}$.  
Therefore, for any ideal $I$ of the algebra $\twl$, 
if $\twl'_{-}\cap I\neq 0$ then $\twl_{+}\cap I\neq 0$. 
Since $\ker\vf$ is an ideal of $\twl$ and $\twl_{+}\cap\ker\vf=0$, 
we obtain $\twl'_{-}\cap\ker\vf=0$.
 
The properties $\twl'_{-}\cap\ker\vf=0$ and $\vf(\twl'_{-})=\sla$ imply 
that $\vf\big|_{\twl'_{-}}$ determines an isomorphism between $\twl'_{-}$ and $\sla$.
As has been shown above, we have $\vf\big(\la^{-1}(\br-\frac1n\tr(\br)\bE)\big)=-S_{\br}$.
\end{proof}

\section*{Acknowledgements}

The authors would like to thank S.~Shadrin for useful discussions. 
The work of SI was supported by 
the Netherlands Organisation for Scientific Research (NWO) grants 613.000.906 and 639.031.515.
MM gratefully acknowledges the support by the
Czech Science Foundation (GA\v{C}R) under project P201/12/G028.


\begin{thebibliography}{99}

\bibitem{akns}
M.J. Ablowitz, D.J. Kaup, A.C. Newell, and H. Segur, 
The inverse scattering transform -- Fourier analysis for nonlinear problems, 
{\it Studies in Appl. Math.} {\bf 53} (1974) 249--315.

\bibitem{Am}
Yu. Aminov, On immersions of regions of the $n$-dimensional Lobachevsky space 
into $(2n - 1)$-dimensional Euclidean space, {\it Dokl. Akad. Nauk SSSR} 
{\bf 236} (1977) 521--524 (in Russian); English transl. {\it Sov. Math. Dokl.} 
{\bf 18} (1977) 1210-1213.

\bibitem{An}
S.C. Anco,
Hamiltonian flows of curves in $G/SO(N)$ and vector 
soliton equations of mKdV and sine-Gordon type,  
{\it SIGMA Symmetry Integrability Geom. Methods Appl.} {\bf 2} (2006), Paper 044, 18 pages. 
 
\bibitem{aglz}
H.~Aratyn and J.~van de Leur, 
Solutions of the WDVV equations and integrable hierarchies of KP type, 
\emph{Comm. Math. Phys.} \textbf{239} (2003) 155--182. 

\bibitem{A-F}
C. Athorne and A. Fordy,
Generalised KdV and MKdV equations associated with symmetric spaces,
{\it J. Phys. A: Math. Gen.} {\bf 20} (1987) 1377--1386.

\bibitem{B-M}
M.Yu. Balakhnev and A.G. Meshkov,
On a classification of integrable vectorial evolutionary equations,
{\it J. Nonlin. Math. Phys.} {\bf 15} (2008) 212--226.

\bibitem{vinbook} 
A.V. Bocharov, V.N. Chetverikov, S.V. Duzhin, N.G. Khor{\cprime}kova,
I.S. Krasil{\cprime}shchik, A.V.~Samokhin, Yu.N.~Torkhov, 
A.M.~Verbovetsky, and  A.M.~Vinogradov.
{\it Symmetries and Conservation Laws for Differential Equations of
Mathematical Physics},
American Mathematical Society, Providence, RI, 1999.

\bibitem{bur-shadr}
A.~Buryak and S.~Shadrin,
A remark on deformations of Hurwitz Frobenius manifolds,  
\emph{Lett. Math. Phys.} \textbf{93} (2010) 243--252.

\bibitem{Dar}
G.~Darboux,
Le\c cons sur les syst\`emes orthogonaux et les coordonn\'ees curvilignes,
Gauthier-Villars, Paris, 1910.

\bibitem{dubrovin90}
B.A.~Dubrovin,
Differential geometry of strongly integrable systems of hydrodynamic type, 
\emph{Funct. Anal. Appl.} \textbf{24} (1990) 280--285.
 
\bibitem{dubrovin92}
B.~Dubrovin, Integrable systems in topological field theory,
\emph{Nuclear Phys. B} \textbf{379} (1992) 627--689.
	
\bibitem{dubrovin96}
B.~Dubrovin, Geometry of $2D$ topological field theories, in:
\emph{Integrable systems and quantum groups (Montecatini Terme, 1993)}, 120--348, 
Lecture Notes in Math., 1620, Springer, Berlin, 1996.
	
\bibitem{Ego}
D.F.~Egoroff,
A Class of Orthogonal Systems, \emph{Uch. Zap. Mosk. Univ. Otd. Fiz.-Mat.} \textbf{18} (1901) 
1--239.

\bibitem{fvdls}
E.~Feigin, J.~van de Leur, and S.~Shadrin, 
Givental symmetries of Frobenius manifolds and multi-component KP tau-functions,
\emph{Adv. Math.} \textbf{224} (2010) 1031--1056. 

\bibitem{fokas87}
A.S.~Fokas, Symmetries and integrability, \emph{Stud. Appl. Math.} 
\textbf{77} (1987) 253--299. 

\bibitem{gerdjikov}
V.S.~Gerdjikov, G.~Vilasi, and A.B.~Yanovski,  
\emph{Integrable Hamiltonian hierarchies. Spectral and geometric methods},  
Lecture Notes in Physics, 748. Springer-Verlag, Berlin, 2008.

\bibitem{givental1}
A.B.~Givental, Semisimple Frobenius structures at higher genus, 
\emph{Int. Math. Res. Notices} \textbf{2001}, no. 23, 1265--1286.

\bibitem{givental2}
A.B.~Givental, Gromov-Witten invariants and quantization of quadratic Hamiltonians,
\emph{Mosc. Math. J.} \textbf{1} (2001) 551--568.

\bibitem{G-K-S}
M.~G\"urses, A.~Karasu, and V.V.~Sokolov,
On construction of recursion operators from Lax representation,
{\it J. Math. Phys.} {\bf 40} (1999) 6473--6490.

\bibitem{G}
G.A.~Guthrie,
Recursion operators and non-local symmetries,
{\it Proc. R. Soc. London A} {\bf 446} (1994) 107--114.

\bibitem{guth-hick}
G.A.~Guthrie and M.S.~Hickman, 
Nonlocal symmetries of the KdV equation,
{\it J. Math. Phys.} {\bf 34} (1993) 193--205. 

\bibitem{K}
B.G.~Konopelchenko, {\it Nonlinear Integrable
Equations. Recursion Operators, Group-Theoretical and Hamiltonian
Structures of Soliton Equations},
Lecture Notes in Physics, 270. 
Springer-Verlag, Berlin, 1987.

\bibitem{kerst-kras}
I.S.~Krasil{\cprime}shchik and P.H.M.~Kersten, 
\emph{Symmetries and recursion operators for classical and supersymmetric differential equations}, Kluwer Academic Publishers, Dordrecht, 2000.

\bibitem{K-V1}
I.S.~Krasilshchik and A.M.~Vinogradov,
Nonlocal symmetries and the theory of coverings: An addendum to A.M.
Vinogradov's `Local symmetries and conservation laws',
{\it Acta Appl. Math.} {\bf 2} (1984) 79--96.

\bibitem{K-V2}
I.S.~Krasil{\cprime}shchik and A.M.~Vinogradov,
Nonlocal trends in the geometry of differential equations: symmetries,
conservation laws, and B\"acklund transformations,
{\it Acta Appl. Math.} {\bf 15} (1989) 161--209.

\bibitem{krich}
I.M.~Krichever, 
Algebraic-geometric $n$-orthogonal curvilinear coordinate systems and solutions 
of the associativity equations, \emph{Funct. Anal. Appl.} \textbf{31} (1997) 25--39. 

\bibitem{lee1}
Y.-P.~Lee, Invariance of tautological equations. I. Conjectures and applications. 
\emph{J. Eur. Math. Soc. (JEMS)} \textbf{10} (2008) 399--413.

\bibitem{lee2}
Y.-P.~Lee, Invariance of tautological equations. II. Gromov-Witten theory. 
\emph{J. Amer. Math. Soc.} \textbf{22} (2009) 331--352.

\bibitem{vdltw}
J.~van de Leur, Twisted ${\rm GL}_n$ loop group orbit and solutions of the WDVV equations, 
\emph{Int. Math. Res. Notices} \textbf{2001}, no. 11, 551--573.

\bibitem{vdlmar}
J.W.~van de Leur and R.~Martini, 
The construction of Frobenius manifolds from KP tau-functions, 
\emph{Comm. Math. Phys.} \textbf{205} (1999) 587--616.

\bibitem{manin-book}
Yu.I.~Manin, 
\emph{Frobenius manifolds, quantum cohomology, and moduli spaces},
American Mathematical Society, Providence, RI, 1999.

\bibitem{marvan-zcr92}
M.~Marvan, On zero-curvature representations of partial differential equations, in:
\emph{Differential geometry and its applications (Opava, 1992)}, 103--122, 
Math. Publ., 1, Silesian Univ. Opava, Opava, 1993; 
ELibEMS http://www.emis.de/proceedings

\bibitem{M2}
M.~Marvan,
Another look on recursion operators,
in: {\it Differential Geometry and Applications},
Proc. Conf. Brno, 1995 (Masaryk University, Brno, 1987) 393--402;
ELibEMS http://www.emis.de/proceedings

\bibitem{M3}
M.~Marvan, Reducibility of zero curvature representations with application 
to recursion operators, {\it Acta Appl. Math.} {\bf 83} (2004) 39--68.

\bibitem{M-Einstein}
M.~Marvan, Recursion operators for vacuum Einstein equations with symmetries, 
in: {\it Symmetry in nonlinear mathematical physics (Kyiv, 2003). Part 1}, 179--183, 
Proceedings of the Institute of Mathematics of the National Academy of Sciences of Ukraine, 
\textbf{50} (2004) 179--183;  
arXiv:nlin/0401014

\bibitem{M-P}
M.~Marvan and M. Pobo\v{r}il,
Recursion operator for the intrinsic generalized sine-Gordon equation, 
{\it Fundam. Prikl. Mat.} {\bf 12} (2006) (7) 117--128 (in Russian); English translation:
{\it J. Math. Sci. (N. Y.)} {\bf 151} (2008) 3151--3158; 
arXiv:nlin/0605015

\bibitem{M-S}
M. Marvan and A. Sergyeyev,
Recursion operator for the Nizhnik--Veselov--Novikov equation,
{\it J. Phys. A: Math. Gen.} {\bf 36} (2003) L87--L92.

\bibitem{Ol}
P.J.~Olver,
Evolution equations possessing infinitely many symmetries,
{\it J. Math. Phys.} {\bf 18} (1977) 1212--1215.

\bibitem{olv_eng2}
P.J.~Olver, 
\emph{Applications of Lie groups to differential equations}, 
Springer-Verlag, New York, 1993.

\bibitem{O-S}
P.J. Olver and V.V. Sokolov,
Integrable evolution equations on associative algebras,
{\it Comm. Math. Phys.} {\bf 193} (1998) 245--268.

\bibitem{Pap}
C.J. Papachristou,
Lax pair, hidden symmetries, and infinite sequences of conserved currents for self-dual Yang--Mills fields,
{\it J. Phys. A: Math. Gen.} {\bf 24} (1991) L1051--L1055.

\bibitem{sakovich-zcr}
S.Yu.~Sakovich, 
Cyclic bases of zero-curvature representations: five illustrations to one concept, 
{\it Acta Appl. Math.} {\bf 83} (2004) 69--83. 

\bibitem{S-W}
J.A. Sanders and J.P. Wang, 
On recursion operators, {\it Physica D} {\bf 149} (2001) 1--10.

\bibitem{S-S}
V.V. Sokolov and S.I. Svinolupov, Vector-matrix generalizations of classical
integrable equations, {\it Teor. Mat. Fiz.} {\bf 100} (1994) 959--962.

\bibitem{T-T}
K. Tenenblat and Chuu Lian Terng, A higher dimension generalization of the
sine-Gordon equation and its B\"acklund transformation,
{\it Bull. Amer. Math. Soc. (N.S.)} {\bf 1} (1979) 589--593.

\bibitem{CLT}
Chuu Lian Terng,
A higher dimension generalization of the sine-Gordon equation and its
soliton theory.
{\it Ann. of Math.} {\bf 111} (1980) 491--510.

\bibitem{tsarev}
S.P.~Tsarev,
Classical differential geometry and integrability of systems of hydrodynamic type. 
\emph{Applications of analytic and geometric methods to nonlinear differential equations 
(Exeter, 1992)}, 241--249,
NATO Adv. Sci. Inst. Ser. C Math. Phys. Sci., 413, Kluwer Acad. Publ., Dordrecht, 1993; 
arXiv:hep-th/9303092 

\bibitem{we_prol1975}
H.D.~Wahlquist and F.B.~Estabrook,  
Prolongation structures of nonlinear evolution equations, 
\emph{J. Math. Phys.} \textbf{16} (1975), 1--7.

\bibitem{zakhduke}
V.E.~Zakharov, 
Description of the $n$-orthogonal curvilinear coordinate systems and 
Hamiltonian integrable systems of hydrodynamic type. I. Integration of the Lam\'e equations, 
\emph{Duke Math. J.} \textbf{94} (1998) 103--139. 

\end{thebibliography}
\end{document}